\documentclass{amsart}%
\usepackage{amsfonts}%
\usepackage{amsmath}%
\usepackage{amssymb}%
\usepackage{graphicx}
\usepackage{geometry}
\usepackage{enumerate}
\usepackage{color}
\usepackage[table]{xcolor}
\usepackage{algorithm}
\usepackage{multirow}

\usepackage{enumitem}
\providecommand{\U}[1]{\protect \rule{.1in}{.1in}}
\newtheorem{theorem}{Theorem}
\theoremstyle{plain}

\newtheorem{definition}{Definition}

\newtheorem{remark}{Remark}

\numberwithin{equation}{section}

\begin{document}
\title[Short Title]{
    Solar Panel Selection using Extended WASPAS with Disc Intuitionistic Fuzzy Choquet Integral Operators: CASPAS Methodology
}
\author{Mahmut Can Bozy\.{ı}\u{g}\.{ı}t$^{1}$}
	\address{$^{1}$Ankara Y{\i}ld{\i}r{\i}m Beyaz{\i}t University, Faculty of Engineering and Natural Sciences, Department of Mathematics, 	06010 Ankara Turkey} \email{mcbozyigit@aybu.edu.tr}
\author{Mehmet \"{u}nver$^{2}$}
	\address{$^{2}$Ankara University, Faculty of Science, Department of Mathematics, 06100 Ankara Turkey}
	\email{munver@ankara.edu.tr}
\subjclass[2020]{94D05, 90B50, 28E10}
\keywords{D-IFSs, Choquet integral operators, CASPAS, MCDM, Solar panels}
\begin{abstract}
Renewable energy is crucial for addressing the growing energy demands of modern society while mitigating the adverse effects of climate change. Unlike fossil fuels, renewable energy sources such as solar, wind, hydro, geothermal, and biomass are abundant, sustainable, and environmentally friendly. This study focuses on addressing a critical challenge in renewable energy decision-making by developing a novel framework for optimal solar panel selection, a key component of sustainable energy solutions. Solar panel selection involves evaluating multiple interdependent criteria, such as efficiency, cost, durability, and environmental impact. Traditional multi-criteria decision-making (MCDM) methods often fail to account for the interdependencies among these criteria, leading to suboptimal outcomes. To overcome this limitation, the study introduces the Choquet Aggregated Sum Product Assessment (CASPAS) method, a Choquet integral-based MCDM approach that incorporates fuzzy measures to model interactions among criteria. CASPAS generalizes the Weighted Aggregated Sum Product Assessment (WASPAS) method, thereby enhancing decision-making accuracy and reliability. This study also introduces the concept of disc intuitionistic fuzzy set (D-IFS), a generalization of the concept of circular intuitionistic fuzzy set, which employ a radius function capable of assigning varying values to individual elements instead of relying on a fixed radius.  Recognizing that traditional weighted aggregation operators neglect the interaction among criteria, this study proposes disc intuitionistic fuzzy Choquet integral operators by incorporating the concept of fuzzy measures, which are effective in modeling such interactions. The proposed method is applied to a renewable energy problem on selecting optimal solar panels. This application demonstrates how the CASPAS method develops decision-making accuracy by incorporating interdependencies among selection criteria. Comparative analysis with existing MCDM methods, alongside sensitivity and validation analyses, underscore the robustness and effectiveness of the CASPAS method.

\end{abstract}
\maketitle

\section{Intruduction}

Renewable energy sources such as solar, wind, hydroelectric, geothermal, and biomass are essential for achieving a sustainable future. These sources offer significant benefits by reducing greenhouse gas emissions and helping to mitigate climate change \cite{Sathaye, Sorensen}. 
Figure \ref{fig-a} shows all types of renewable energy. Solar energy captures sunlight using photovoltaic cells, while wind energy converts the kinetic energy of wind into electricity through turbines. Hydroelectric power generates electricity by utilizing the potential energy stored in water behind dams. Geothermal energy taps into the heat from the Earth's interior, and biomass energy comes from organic materials. These sources are environmentally friendly, helping to reduce our dependence on fossil fuels, which are major contributors to carbon dioxide emissions and global warming \cite{Attanayake}. Although renewable energy technologies often come with high initial costs and can be influenced by geographic and climatic conditions, ongoing advancements in technology and supportive policies are making these sources more accessible and cost-effective \cite{Dincer}. Transitioning to renewable energy not only addresses environmental challenges but also promotes energy security and stimulates economic growth \cite{Tekbıyık}.

\begin{figure}[h!]
    \centering
    \includegraphics[scale=0.42]{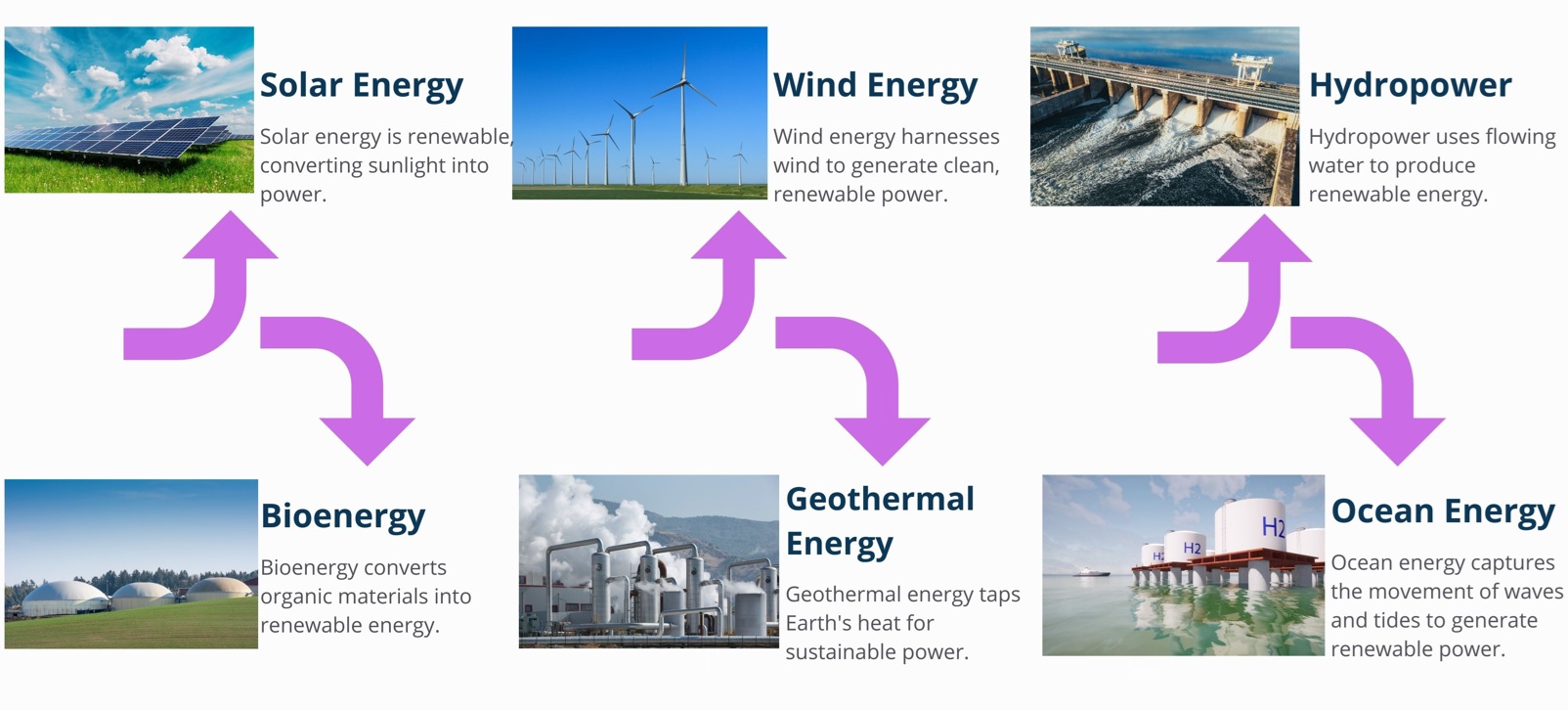}
\caption{Types of renewable energy}
    \label{fig-a}
\end{figure}

\footnotesize

\begin{table}[ht]
\centering

\begin{tabular}{llrrrrr}
\hline
\textbf{Region} & \textbf{Year}& \textbf{Hydropower} & \textbf{Wind E.} & \textbf{Solar E.} & \textbf{Bioenergy} & \textbf{Geothermal E.} \\
\hline
      & 2023   & 1 406 863  & 1 017 390  & 1 418 016  & 148 840  & 15 026  \\
World   & 2022       & 1 394 824  & 902 883  & 1 070 851  & 144 286  & 14 653  \\
        & 2021  & 1 362 416  & 824 321  & 870 643  & 138 338  & 14 432  \\ \hline
      & 2023   & 40 280     & 8 654      & 13 438     & 1 899   & 991     \\
Africa  & 2022       & 39 275    & 7 745      & 12 646     & 1 916   & 956    \\
    & 2021     & 37 529     & 6 909     & 11 582     & 1 917   & 870    \\ \hline
      & 2023     & 630 030    & 508 599    & 839 329    & 62 068  & 5 011  \\
Asia  & 2022          & 620 391   & 426 482    & 601 390    & 58 729  & 4 761   \\
  & 2021         & 594 325    & 384 742    & 488 907    & 54 608  & 4 732   \\ \hline
     & 2023    & 225 728   & 255 615    & 288 644    & 42 825  & 1 637  \\
Europe  & 2022        & 225 804    & 240 278    & 233 906    & 42 458  & 1 637  \\
   & 2021      & 224 212    & 221 517    & 190 928    & 41 601  & 1 633   \\ \hline

  & 2023 & 199 492    & 172 328    & 156 000    & 14 576  & 3 678   \\
N. America & 2022   & 199 581    & 164 389    & 129 156    & 14 684 & 3 647   \\
 & 2021  & 198 861    & 154 210    & 109 579    & 14 906  & 3 595  \\ \hline
 & 2023 & 180 962    & 39 852    & 49 392    & 20 348  & 83      \\
S. America & 2022 & 180 018    & 33 580    & 34 698   & 19 552 & 51      \\
 & 2021 & 178 482   & 29 737     & 21 259    & 18 622  & 40     \\ \hline
     &2023  & 15 542     & 14 045     & 33 417     & 1 106   & 1 100   \\
Oceania  &2022      & 14 749     & 13 035     & 29 566     & 1 104   & 1 100   \\
     &2021    & 14 509     & 11 522     & 24 895     & 1 104  & 1 093   \\
\hline
\end{tabular}
\caption{CAP (MW) by energy source and region between 2021 and 2023 according to IRENA \cite{IRENA}}
\label{cap}
\end{table}
\normalsize

Solar energy offers several key advantages over other renewable sources. One of its primary benefits is its accessibility and flexibility. Solar panels can be installed on rooftops or open spaces, making it suitable for both urban and rural areas \cite{Novas}. Unlike wind turbines, which require specific locations with consistent wind or hydroelectric plants that need large water bodies and can disrupt local ecosystems, solar installations are less intrusive and can be deployed in diverse environments \cite{Lazaroiu,Sathaye}. Moreover, solar energy systems have low maintenance costs and produce no emissions during operation, making them an environmentally friendly choice \cite{Novas}. The falling costs of solar technology also make it an increasingly affordable option for many consumers \cite{Lazaroiu,Sathaye}. Table \ref{cap} presents the global capacity in megawatt (MW) for renewable energy sources from 2021 to 2023 \cite{IRENA}. Table \ref{cap} highlights global and regional trends in renewable energy capacities (MW) from 2021 to 2023, showing significant growth in solar and wind energy. Solar energy, the fastest-growing source, nearly doubles globally, with Asia leading the increase, followed by Europe and North America. This rapid growth reflects falling costs of photovoltaic technology, increasing efficiency, and widespread policy support for clean energy transitions. Wind energy also expands rapidly, particularly in Asia and Europe, while hydropower, the largest contributor, grows more modestly due to its already extensive development. Regions like South America and Africa exhibit slower but steady progress, focusing on solar and hydropower. Bioenergy and geothermal energy maintain stable contributions worldwide. The remarkable rise of solar energy not only emphasizes its critical role in reducing carbon emissions but also highlights its accessibility and scalability, making it a cornerstone of the global shift toward cleaner, more sustainable energy systems. Figure \ref{fig-54} shows the growth of solar photovoltaic  production in Gigawatt-hours (GWh) globally and across regions from 2014 to 2022 \cite{IRENA}. Worldwide production increased significantly, from 183,473 GWh in 2014 to 1,281,654 GWh in 2022, reflecting rapid adoption of solar energy. Asia leads with the highest production, growing from 55,392 GWh to 686,755 GWh, followed by Europe and North America, which reached 231,064 GWh and 210,107 GWh, respectively. While Africa, South America, and Oceania started with lower production levels, they also saw notable growth, particularly after 2019, driven by investments in renewable energy and technological advancements. This trend underscores the global shift toward cleaner energy sources.

\begin{figure}[h!]
    \centering
    \includegraphics[scale=0.55]{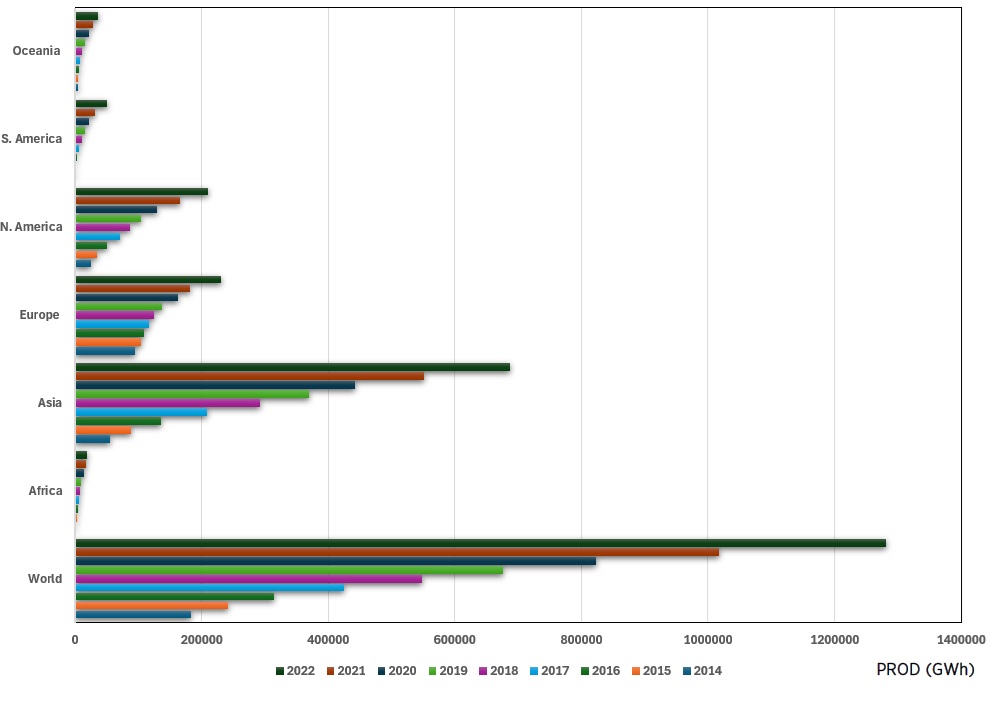}
\caption{Global solar photovoltaic production data by region between 2014 and 2022}
    \label{fig-54}
\end{figure}

Solar panels, often referred to as photovoltaic panels, are systems that harness sunlight and transform it into electricity \cite{Fraas}. This process, known as the photovoltaic effect, occurs when light energy stimulates electrons within a semiconductor material, producing an electric current \cite{Nag}. Made up of numerous interconnected solar cells, these panels are utilized extensively across residential, commercial, and industrial sectors, as well as in space missions \cite{Spanggaard}. This technology, which has an important place among renewable energy sources today, has undergone a long scientific and technological development process. The roots of solar panel technology trace back to 1839, when Becquerel \cite{Becquerel} discovered the photovoltaic effect. He found that certain materials could produce an electric current under sunlight \cite{Fraas}. While this discovery was revolutionary, it initially remained a theoretical concept without immediate practical use. Becquerel's observation went unreplicated until 1873, when Smith \cite{Smith} found that light striking selenium could generate a charge. Building on this finding, Adams and Day \cite{Adams} conducted experiments to confirm Smith's results and published their work, The Action of Light on Selenium, in 1876. In 1883, Charles Fritts developed the first functional solar cell using selenium \cite{Fraas,Nag}. Despite its low efficiency of just 1\%, this marked the beginning of the technology capable of directly converting sunlight into electricity \cite{Fraas}.  Einstein \cite{Einstein} made a significant impact on solar panel technology by explaining the photoelectric effect in 1905. He proposed that light consists of tiny energy packets known as photons. When these photons hit a material, they can release electrons from its surface, a process known as the photoelectric effect. This discovery is fundamental to the operation of solar cells \cite{Nag}. In 1939,  Ohl \cite{Ohl} developed the solar cell design that is now commonly used in many modern solar panels. He patented his design in 1941. In 1954, scientists  Chapin et al. \cite{Chapin} at Bell Laboratories developed the first modern solar cell using silicon. Their creation had an efficiency of around 6\%, marking the start of practical solar panel technology \cite{Fraas}. During this period, solar panels became commercially available, but their high production costs restricted their use primarily to specialized applications like satellites. Solar panels were first used in space on the Vanguard I satellite in 1958. This marked a critical point for solar energy development, as it demonstrated their reliability in extreme conditions \cite{Easton}. The 1970s oil crisis spurred interest in renewable energy, leading to significant investments in solar technology \cite{Schumacher}. During this time, the cost of solar power began to decline as production scaled up. Today, solar panels play a crucial role in global renewable energy strategies, with ongoing innovations continually improving their efficiency and making them more accessible \cite{Nag}. Figure \ref{timeline} illustrates the timeline of solar panel development.

\begin{figure}[h!]
    \centering
    \includegraphics[scale=0.37]{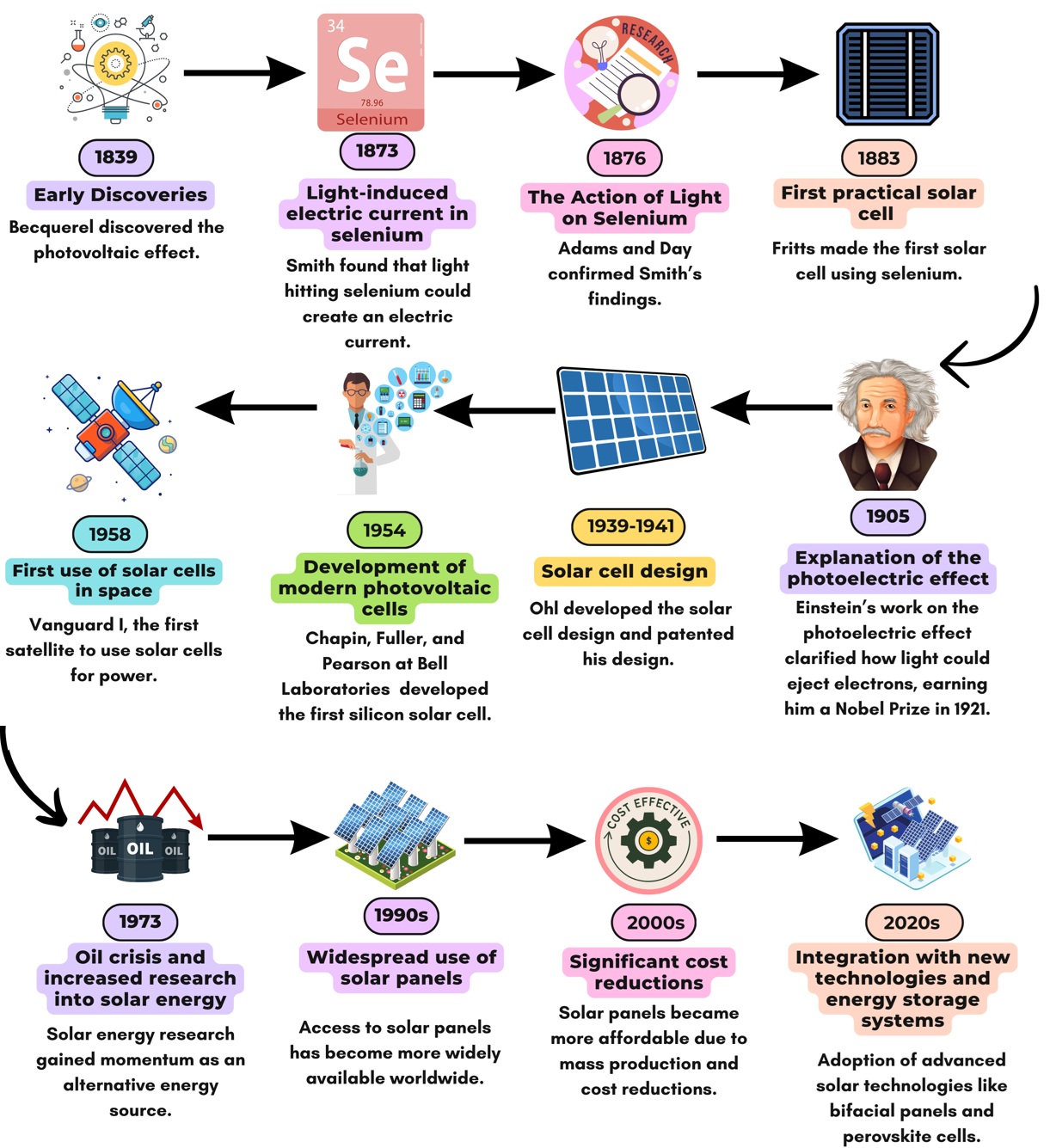}
\caption{Timeline of the development of solar panels}
    \label{timeline}
\end{figure}

Fuzzy sets (FSs), introduced by Zadeh \cite{Zadeh}, represent a significant extension of classical set theory by providing a means to handle uncertainty and vagueness in data. Unlike crisp sets, where an element either fully belongs to a set or does not belong at all, FSs allow for partial membership. Formally, a FS \( A \) in a universe of discourse \( \Delta \) is defined by a membership function \( \mu_A: \Delta \to [0,1] \), where \( \mu_A(\delta) \) quantifies the degree to which an element \( \delta\) belongs to \( A \). This approach is particularly useful in scenarios where sharp boundaries are not appropriate, such as in the analysis of imprecise or subjective information. By accommodating varying degrees of membership, FSs provide a flexible mathematical framework for modeling real-world phenomena where crisp distinctions are not feasible. 

Despite their utility, FSs have limitations in capturing the full spectrum of uncertainty. To address these, Atanassov \cite{Atanassov1} proposed intuitionistic fuzzy sets (IFSs), which further improve the expressive power of FSs. IFSs expand on the concept of partial membership by associating each element in a set with two  parameters: a membership degree and a non-membership degree. These values provide a more nuanced representation of uncertainty by allowing the sum of membership and non-membership degrees to be less than or equal to 1, thereby introducing a third dimension called hesitation. Hesitation reflects the degree of indeterminacy or lack of knowledge about an element’s status within the set. This additional flexibility makes IFSs highly effective for applications requiring a deeper and more comprehensive representation of uncertainty, such as decision-making, pattern recognition, and fuzzy control systems. By capturing both membership and non-membership alongside hesitation, IFSs offer a more robust tool for tackling complex problems characterized by ambiguity and incomplete information.

In decision-making problems involving uncertainty and ambiguity, assigning specific intuitionistic fuzzy values (IFVs) to evaluation criteria can be challenging. This difficulty stems from the inherent uncertainty in determining the exact degrees to which an element belongs or does not belong to a set. Traditional IFSs may lack the necessary flexibility to handle situations where multiple decision-makers interpret criteria differently or have diverse perspectives. To address this limitation, Atanassov \cite{Atanassov2, Atanassov3} introduced the concept of circular intuitionistic fuzzy sets (C-IFSs), an extension of the conventional IFS framework. C-IFSs use a circular structure to represent the degrees of membership and non-membership, which inherently captures the vagueness and uncertainty associated with these values. This circular representation provides a more dynamic and sensitive way to reflect changes in both membership and non-membership degrees, making it particularly useful in complex decision-making scenarios. The circular design enhances flexibility by accommodating the imprecision and ambiguity often present in real-world problems. As a result, C-IFSs offer decision-makers a more robust framework for evaluating alternatives, especially when assigning precise IFVs to criteria is impractical. 
Building on this concept, Bozyiğit et al. \cite{Bozyiğit} introduced circular Pythagorean fuzzy sets (C-PFSs), which expand upon C-IFSs by offering a broader representational range. A C-PFS is visualized as a circle that encapsulates both the degree of membership and the degree of non-membership. The center of this circle is defined by two non-negative real numbers, \( \mu \) and \( \nu \), which must satisfy the condition \( 0 \leq \mu^2 + \nu^2 \leq 1 \). Since the radius degree remains constant for each C-PFS, disc Pythagorean fuzzy sets (D-PFSs) were proposed by replacing this constant radius with a radius function \cite{Khan}. This approach introduces flexibility to the radius degree.

Multi-Criteria Decision-Making (MCDM) is a specialized field within decision science that addresses the challenge of identifying the most suitable alternatives when faced with multiple and sometimes conflicting criteria. It offers a structured approach to help decision-makers navigate complex situations by systematically analyzing trade-offs between various objectives. MCDM techniques combine both quantitative and qualitative information, enabling a comprehensive evaluation of options. Through the use of mathematical frameworks and well-defined decision processes, MCDM improves the ability to pinpoint optimal solutions, ensures greater clarity and fairness in decision-making, and supports alignment with overarching strategic objectives. Weighted Aggregated Sum Product Assessment (WASPAS) method, introduced by Zavadskas et al. \cite{Zavadskas1}, is a robust and versatile MCDM approach that integrates  Weighted Sum Model (WSM) and  Weighted Product Model (WPM). By combining these two methodologies, WASPAS aims to develop decision-making accuracy and reliability. It assigns weights to criteria based on their relative importance and evaluates alternatives through a hybrid aggregation mechanism, which includes both additive and multiplicative components. The flexibility of WASPAS allows it to effectively handle a wide range of decision-making scenarios, particularly in cases where decision-makers need to balance between conflicting criteria. Fuzzy WASPAS extends the traditional WASPAS method by incorporating fuzzy set theory to address uncertainties and imprecision in decision-making \cite{Turskis, Zavadskas2}. In real-world applications, decision-makers often face ambiguity in assigning precise values to criteria weights and alternatives. Fuzzy WASPAS employs fuzzy values to capture this vagueness. This enables a more realistic representation of preferences and better models human judgment. The fuzzy extension makes the WASPAS method particularly suitable for complex decision-making environments, such as supply chain management, project selection, and resource allocation, where uncertainty is a significant factor. Using the strengths of both fuzzy logic and the hybrid aggregation mechanism, Fuzzy WASPAS provides a powerful tool for solving intricate MCDM problems.

Choquet \cite{Choquet} introduced the concept of capacity and the Choquet integral, an extension of the Lebesgue integral. When defined using an additive measure, Choquet integral simplifies to a weighted arithmetic mean. Building on this, Sugeno \cite{Sugeno} expanded the notion of capacity to fuzzy measure and developed the Sugeno integral. Unlike additive measures, fuzzy measures  can better capture the interaction among criteria, making fuzzy integrals particularly effective in such contexts. Fuzzy measures and fuzzy integrals are frequently used tools in decision making and pattern recognition problems (\cite{Olgun, Unver,Unver2}). Among various fuzzy measures, the $\lambda$-fuzzy measure, defined by Sugeno \cite{Sugeno}, is notably significant. This measure is constructed using interaction indices and element weights. Takahagi \cite{Takahagi1} proposed three algorithms to efficiently generate $\lambda$-fuzzy measure, accounting for both interaction indices and weight allocations. Further extending these ideas, Tan and Chen \cite{Tan} introduced the IF Choquet integral operator, designed to aggregate IFVs for  MCDM in place of traditional operators reliant on additive measures. Later, Peng and Yang \cite{PengY} generalized the Choquet integral operator from IFVs to Pythagorean fuzzy values (PFVs), presenting MABAC method based on PF Choquet integral operator. Bozyiğit et al. \cite{Bozyiğit2} proposed C-PF Choquet integral operators, which generalize some weighted arithmetic aggregation operators, to aggregate circular Pythagorean fuzzy values (C-PFVs).

This paper offers four significant contributions:

\begin{itemize}
    \item  \textbf{Theoretical contributions:} To improve the flexibility of the fixed radius degree in C-IFSs, this paper introduces the concept of disc intuitionistic fuzzy sets (D-IFSs), which incorporate a radius function. D-IFSs generalize C-IFSs by providing a more expansive framework. Following this, set operations specific to D-IFSs are introduced. Additionally, algebraic operations on disc intuitionistic fuzzy values (D-IFVs) are presented to extend the radius degree \(\sqrt{2}\), which is restricted to 1 in circular intuitionistic fuzzy values (C-IFVs) that overlap with D-IFVs. Weighted aggregation operators are then introduced for aggregating D-IFVs. However, since these aggregation operators do not account for interactions between criteria and such interactions may exist in real-world problems, disc intuitionistic fuzzy Choquet integral operators are proposed in the study. These operators use fuzzy measures to consider the interaction between criteria.

\item  \textbf{Methodological Contributions:} CASPAS (Choquet Aggregated Sum Product Assessment) method introduces a novel approach to multi-criteria decision-making by incorporating  Choquet Sum Model (CSM) and Choquet Product Model (CPM), replacing the traditional WSM and WPM used in WASPAS. These innovations employ the Choquet integral to account for interactions and dependencies among criteria, addressing limitations in additive and multiplicative aggregation methods. By integrating a capacity function, CSM and CPM enable more nuanced evaluations, capturing synergies and conflicts between criteria. Compared to WASPAS, CASPAS offers superior accuracy, adaptability, and applicability, marking a significant methodological contribution to decision-making in complex interdependent systems.

\item  \textbf{Applied Contributions:} This study examines a MCDM problem related to the selection of solar panels based on the criteria of efficiency, cost, durability, and installation complexity. The proposed CASPAS method is applied to evaluate some solar panels, demonstrating its capability to handle interdependent criteria effectively. Following the application, comparison analysis, sensitivity analysis, and validity analysis are conducted to assess the method's robustness and reliability. These analyses highlight the practical relevance of CASPAS in solar panel selection and its potential to deliver accurate and consistent decision-making outcomes in real-world applications.

\item  \textbf{Interdisciplinary Contributions:} Using  CASPAS method for solar panel selection, this study brings together multiple disciplines to offer a new approach to decision-making processes. It integrates engineering, energy management, economics, and environmental sciences to effectively evaluate technical, economic, environmental, and installation-related factors in solar panel selection. The use of CASPAS provides an innovative contribution to existing methods in engineering and energy planning by considering the interactions between criteria in MCDM processes. Moreover, this study enables the development of an interdisciplinary approach for sustainable energy solutions and environmental impact assessments. Thus, it allows for a more effective and efficient selection of solar energy systems from both technical engineering and environmental economics perspectives.

\end{itemize}

The structure of the paper is as follows: Section \ref{litre} presents a comprehensive literature review covering C-IFSs, existing extended WASPAS methods, and various types of solar panels. It also introduces the theoretical foundation of the proposed CASPAS method. Section \ref{difs} focuses on the development of fundamental set operations for D-IFSs and the definition of score and accuracy functions for ranking D-IFVs. Additionally, algebraic operations for D-IFVs are formulated by extending the radius limit from 1 to \(\sqrt{2}\). Using these algebraic operations, weighted arithmetic and geometric aggregation operators for D-IFVs are established.  
In Section \ref{choquett}, after recalling the fundamental concepts of fuzzy measures and the Choquet integral, disc intuitionistic fuzzy Choquet interval operators are introduced. These operators extend the weighted arithmetic and geometric aggregation operators by incorporating interactions between criteria, addressing challenges commonly found in real-world problems. Section \ref{appsolar} begins with an overview of solar panel types and subsequently proposes CASPAS method for D-IFSs. This method advances the WASPAS approach by integrating CSM and CPM techniques.  
Section \ref{performance} explores the sensitivity of the CASPAS method to various parameters and provides  validity analysis. Section \ref{comparative} evaluates the performance of CASPAS through comparisons with existing methods in the literature. The paper concludes with Section \ref{confut}, which discusses the findings, draws conclusions, and outlines directions for future research.

\section{Literature research} \label{litre}
This section provides an overview of existing studies and methodologies related to C-IFSs, the types of solar panels and extended WASPAS approaches.

\subsection{A comparison of C-IFSs and D-IFSs}
C-IFSs are an extension of IFS, designed to better express uncertainty and fuzziness. In C-IFSs, each element is represented as a circle, where the center of the circle indicates the membership and non-membership degrees. The radius of the circle represents the hesitation degree. C-IFSs are particularly useful in MCDM, allowing for more flexible modeling of uncertain information. C-IFSs contain more information compared to IFSs, making them effective in solving more complex problems. Atanassov and Marinov  \cite{Atanassov3} introduced the concept of distance measures for C-IFS and proposed four different distance measures based on the classical metric definition. For decision-making applications, Chen \cite{Chen} proposed Minkowski distance measures specifically designed for C-IFSs. Bozyiğit and Ünver \cite{Bozyiğit1} developed score and accuracy functions to rank C-IFVs. They also introduced parametric distance measures based on Hamming and Euclidean metrics for C-IFVs.  Garg et al. \cite{Garg} proposed algebraic operations for C-IFVs, and based on these operations, they introduced weight aggregation operators based on Archimedean t-norms for C-IFVs. Using these aggregation operators, they developed an extended EDAS method, which is a MCDM approach.

In C-IFSs, the radius degree for each element is treated as a constant, which imposes a limitation. To address this restriction, the concept of D-IFSs is introduced in this study, where the radius is variable rather than fixed. In D-IFSs, the radius is represented by a radius function, similar to the membership and non-membership functions. This provides flexibility in the radius degree for application areas such as decision-making, pattern recognition, and classification. Figure \ref{compc-ıfs} displays a visual comparison of C-IFSs and D-IFSs. 
As shown in Figure \ref{compc-ıfs}, the radius remains constant in a C-IFS, whereas in a D-IFS, the radius can vary.

\begin{figure}[h!]
    \centering
    \includegraphics[scale=0.35]{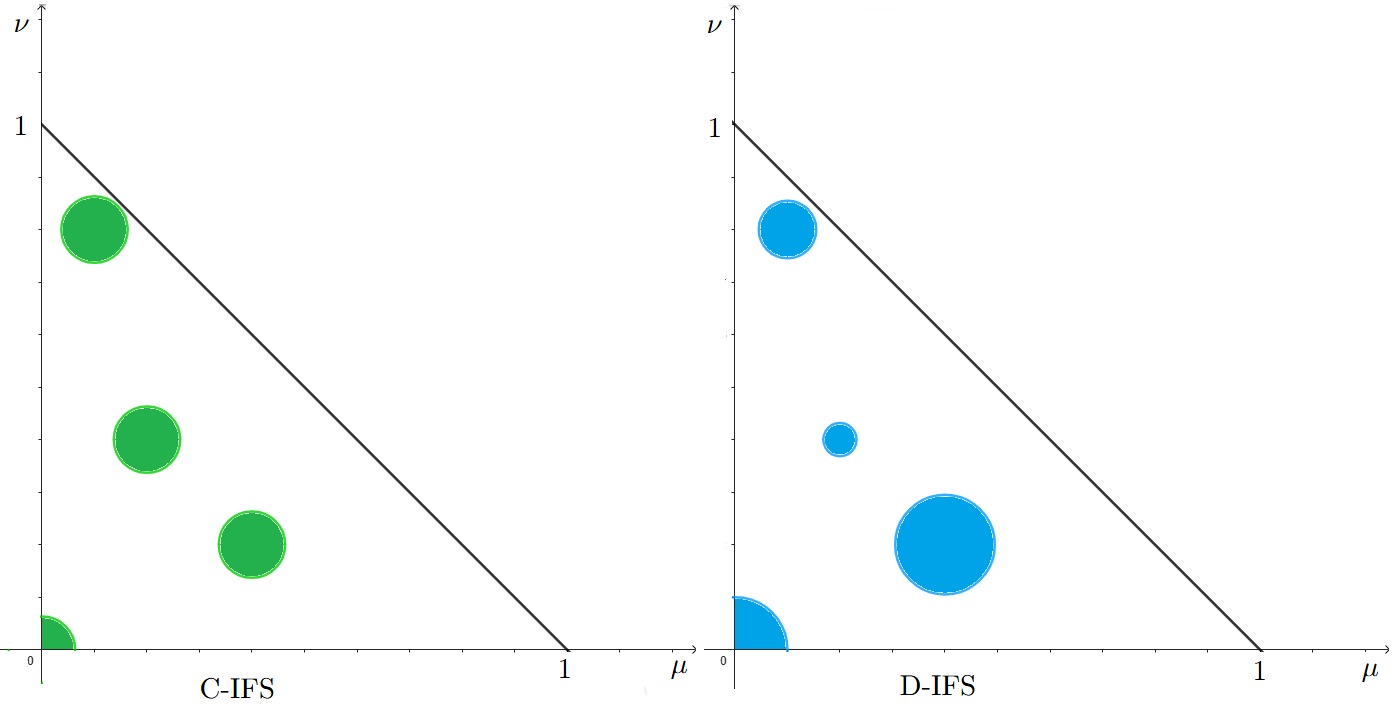}
\caption{Comparison of C-IFSs and D-IFSs}
    \label{compc-ıfs}
\end{figure}

\subsection{A comparison of WASPAS and CASPAS approaches}

 WASPAS method is a well-established MCDM approach that combines the strengths of WSM and WPM. By integrating these two techniques, WASPAS offers a balance between simplicity and computational efficiency. However, a key limitation of WASPAS lies in its assumption of independence among criteria, which may not hold true in real-world scenarios where criteria often exhibit interdependencies. This limitation can lead to suboptimal or less accurate results in complex decision-making problems. WASPAS method has been extensively studied in the literature due to its simplicity and effectiveness in handling MCDM problems. In Table \ref{Tbsas}, a summary of key studies on WASPAS and its applications is provided, illustrating the various areas where the method has been applied and the contributions made in the literature. 
 
 CASPAS, a novel approach, addresses this limitation by employing Choquet integrals through its CSM and CPM. These models allow CASPAS to account for interdependencies among criteria, enabling it to capture the interactions that influence the decision-making process more accurately. By incorporating these interdependencies, CASPAS enhances the realism and reliability of the results, making it particularly suitable for applications with complex criteria relationships.
Furthermore, CASPAS introduces greater flexibility in the aggregation process by providing a mechanism to model the influence of each criterion based on its interaction with others. This feature not only improves the accuracy of the method but also extends its applicability to a broader range of decision-making scenarios. As a result, CASPAS emerges as a superior alternative to WASPAS in contexts where the relationships between criteria play a critical role, such as renewable energy selection, financial analysis, and engineering design problems. Figure \ref{wascas} displays a table comparing the features of WASPAS and CASPAS. Figure \ref{fig-abng} also illustrates the stages and steps involved in the CASPAS method.

\footnotesize
\begin{table}[h]
	\begin{center}
    \renewcommand{\arraystretch}{2}
	\renewcommand{\tabcolsep}{6pt}
	\begin{tabular}
		[c]{lm{4cm}lm{4cm}}\hline
		 \textbf{Year} & \textbf{Environment}& \textbf{Author(s)}  &  \textbf{Application area}
		\\ 
		\hline
2012 & Crips & Zavadskas et al. \cite{Zavadskas1} & Illustrative example \\

2014 & Interval valued intuitionistic fuzzy sets & Zavadskas et al. \cite{Zavadskas2} & Derelict buildings’ redevelopment decisions and investment\\

2015 &  Fuzzy sets & Turskis et al. \cite{Turskis} & Construction site selection\\

2016 &  Interval type-2 fuzzy sets & Keshavarz Ghorabaee et al. \cite{Keshavarz} & Green supplier selection\\

2017 &  Single-valued
neutrosophic set & Baušys and Juodagalvienė \cite{Baušys} & Garage location selection\\

2018 &  Intuitionistic fuzzy sets & Stanujkić and Karabašević \cite{Stanujkić} & Website evaluation\\

2019 &  Spherical fuzzy sets & Kutlu Gundogdu and Kahraman \cite{Kutlu} & Industrial robot selection\\

2020 &   $q$-rung orthopair fuzzy sets & Rani and Mishra \cite{Rani} &  Fuel technology selection\\

2021 &   Fermatean fuzzy  sets &  Mishra and Rani \cite{Mishra} &  Healthcare waste disposal location selection\\

2022 &  Picture fuzzy sets & Senapati and Chen \cite{Senapati} & Air-conditioning system selection\\

2023 &  Pythagorean fuzzy  Sets &  Yalcin Kavus et al.  \cite{Yalcin} & Parcel locker location selection\\

2024 &  Complex Fuzzy Sets & Khan et al. \cite{Khan1} & Parameter Selection Impacting
Software Reliability\\ \hline
  
	\end{tabular}
\caption{Literature survey in different  fuzzy settings  for WASPAS approach} \label{Tbsas}	
\end{center}	
\end{table}
\normalsize

\begin{figure}[h!]
    \centering
    \includegraphics[scale=0.31]{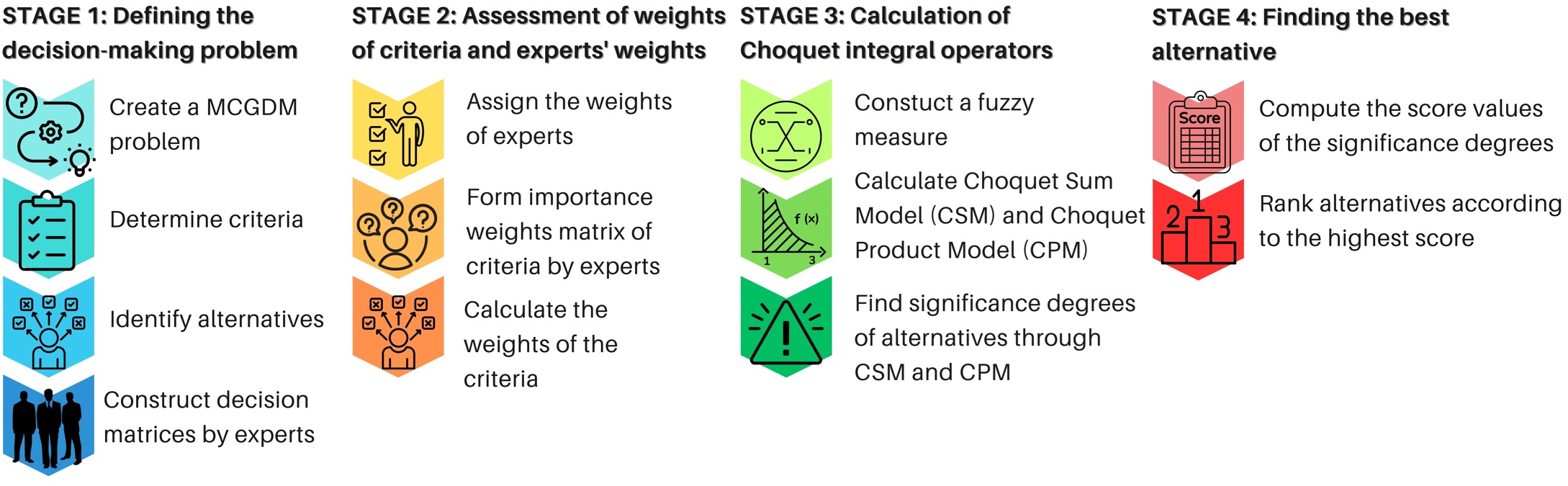}
\caption{Flowchart of  CASPAS method}
    \label{fig-abng}
\end{figure}

\begin{figure}[h!]
    \centering
    \includegraphics[scale=0.31]{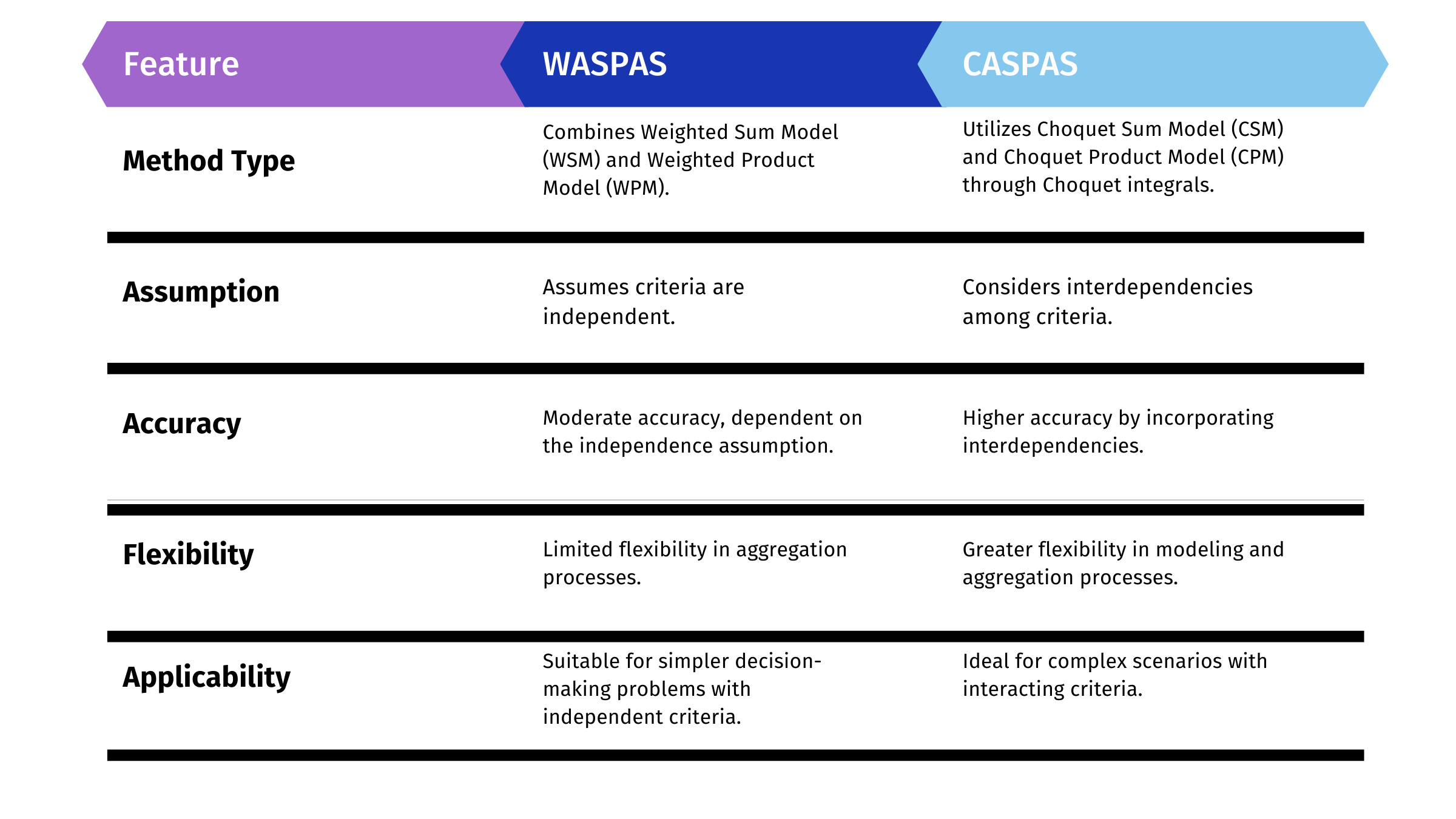}
\caption{Comparison of WASPAS and CASPAS}
    \label{wascas}
\end{figure}

\subsection{Types of solar panels}
Solar panels are advanced systems designed to harness sunlight and convert it into electrical energy, making them a cornerstone of clean and sustainable energy solutions. With rising energy demands and increasing environmental concerns, solar technology has seen significant advancements, resulting in the creation of various panel types tailored to diverse applications \cite{Novas}. Each type offers distinct benefits in terms of efficiency, cost-effectiveness, durability, and suitability for specific uses. The most commonly used types of solar panels are as follows:
\begin{itemize}
\item \textbf{Monocrystalline silicon panels}: They are widely known for their high efficiency, typically ranging from 18\% to 22\%. Made from a single, continuous silicon crystal, these panels are ideal for limited spaces as they generate more energy per square meter. They are more expensive to produce due to their manufacturing process but are highly durable and offer long lifespans, making them a popular choice for residential and commercial installations \cite{Dobrzański, Parida}.

\item \textbf{Polycrystalline silicon panels}:  Polycrystalline panels are less efficient than monocrystalline panels, with efficiency levels around 15\% to 17\%. However, they are more affordable to produce, making them a cost-effective choice for larger installations where space is not a major constraint. These panels are made by melting silicon crystals and forming the panel from the solidified material, which results in a lower efficiency compared to monocrystalline panels \cite{Dobrzański, Parida}.

\item \textbf{Thin-film solar Panels}: Thin-film solar panels are lightweight and flexible, making them ideal for applications where traditional panels might not be suitable. These panels have lower efficiency (typically 10\% to 13\%) but are cheaper to manufacture. Thin-film panels can be used on a variety of surfaces, including curved or irregular structures, and are often used for large-scale projects or specific architectural designs. Thin-film solar cells are preferred due to their low material requirements and improving efficiency levels. The primary technologies in this category are amorphous silicon ($\alpha$-Si), copper indium gallium selenide (CIGS), and cadmium telluride (CdTe) \cite{Lee}.

\item \textbf{PERC (Passivated Emitter and Rear Cell)  solar panels:} PERC panels are an improved version of monocrystalline and polycrystalline panels. They feature an additional layer on the rear side of the cell that reflects unabsorbed light back into the panel, increasing efficiency. This technology boosts energy conversion, particularly in low-light conditions or on cloudy days \cite{Green}.

\item \textbf{Bifacial solar panels:} These panels represent an innovative advancement in photovoltaic technology, capable of generating electricity from both sides of the panel. Unlike conventional monofacial panels, which capture sunlight only on the front, bifacial panels feature solar cells on both the front and rear. This dual-sided design enables them to harness light not only directly from the sun but also from light reflected off the ground or nearby surfaces, greatly enhancing their overall energy production \cite{Guerrero-Lemus}.

\end{itemize}

Each of these panel types offers unique benefits, depending on factors such as material, efficiency, cost, and installation requirements.

Decision-making problems associated with solar energy have been extensively explored in the literature. Notably, the selection of solar panel types and determining optimal locations for solar energy production are among the most significant decision-making issues. Table \ref{solarenergy} provides an overview of studies focusing on decision-making in solar energy.

\footnotesize
\begin{table}[h]
	\begin{center}
    \renewcommand{\arraystretch}{2}
	\renewcommand{\tabcolsep}{6pt}
	\begin{tabular}
		[c]{lm{5cm}lm{4cm}}\hline
		 \textbf{Author(s)} & \textbf{Method}& \textbf{Solar energy application}
		\\ 
		\hline
 Bozyiğit et al. \cite{Bozyiğit} & A MCDM method based on aggregation operators and cosine similarity measures in C-PFSs & Selecting solar cell \\

  El-Bayeh et al. \cite{El-Bayeh} & MCDM algorithm based on the concept of Rank-Weight-Rank& Solar panels selection in buildings\\
  
 Hosouli and Hassani \cite{Hosouli} & Fuzzy Analytic Hierarchy Process method & Solar plant location selection\\

 Kannan et al. \cite{Kannan} &    A hybrid approach based on MCDM methods and Monte Carlo simulation &  Sustainable evaluation of potential solar sites \\
 
 Kaur et al. \cite{Kaur} &   TOPSIS based on entropy technique & Selection of solar panel \\

 Shayani Mehr et al. \cite{Shayani} &  A BWM‐MULTIMOOSRAL approach &  Solar panel technology selection \\

Tüysüz and Kahraman \cite{Tüysüz} &    An integrated picture fuzzy Z-AHP \& TOPSIS methodology & Solar panel selection \\

 \hline 
	\end{tabular}
\caption{Some MCDM studies from the literature on solar energy} \label{solarenergy}	
\end{center}	
\end{table}
\normalsize

\section{Disc Intuitionistic Fuzzy Sets} \label{difs}
Atanassov \cite{Atanassov1} introduced the concept of  IFS, which is an extension of fuzzy sets FS. Throughout this paper $\Delta$ denotes a non-empty finite set.

\begin{definition} \cite{Atanassov1}
 An IFS $I$ in $\Delta$ is defined as:
\[
I = \{\langle \delta_i, \mu_I(\delta_i), \nu_I(\delta_i)\rangle: \delta_i \in \Delta\},
\]
where $\mu_I$ and $\nu_I$ are membership and non-membership functions, respectively with the condition $ \mu_I(\delta_i) + \mu_I(\delta_i) \leq 1$
for every $\delta_i \in \Delta$.
Additionally, an IFV $\theta$ is defined as a pair $\langle\mu_\theta,\nu_\theta \rangle$, where $\mu_\theta$ and $\nu_\theta$ are both within the range $[0, 1]$ and satisfy the condition $\mu_\theta+\nu_\theta\leq 1$. 
\end{definition}

C-IFSs, which was introduced by Atanassov \cite{Atanassov2}, are an extension of IFSs. C-IFSs provide decision makers with a means to express uncertainty by utilizing membership and non-membership degrees represented in the form of a circle. In C-IFSs, the membership and non-membership degrees are visually represented by a circle, allowing decision makers to intuitively understand the degrees of membership and non-membership associated with elements in the set. 

\begin{definition}\cite{Atanassov2,Atanassov3}
    Let $r\in[0,\sqrt{2}]$. A C-IFS $C_r$ in $\Delta$ is defined as:
\[
C_r = \{\langle \delta_i, (\mu_C(\delta_i), \nu_C(\delta_i)); r\rangle: \delta_i \in \Delta\},
\]
where $\mu_C$ and $\nu_C$ are membership and non-membership functions, respectively with the condition $\mu_C(\delta_i) + \nu_C(\delta_i) \leq 1$
for every $\delta_i \in \Delta$. The value $r$ represents the radius of the circle centered at the point $(\mu_C(\delta_i),\nu_C(\delta_i))$ on the plane. This circle visually represents the membership degree and non-membership degree associated with $\delta_i \in \Delta$. 
\end{definition}

\begin{definition}
    A D-IFS $D_r$ in $\Delta$ is defined as:
\[
D_{\hat{r}} = \{\langle \delta_i, (\mu_D(\delta_i), \nu_D(\delta_i)); \hat{r}_D(\delta_i)\rangle: \delta_i \in \Delta\},
\]
where $\mu_D$ and $\nu_D$ are membership and non-membership functions, respectively with the condition $\mu_D(\delta_i) + \nu_D(\delta_i) \leq 1$
for every $\delta_i \in \Delta$ and $r_D$ is  a radius function with $r_D(\delta_i)\in [0,\sqrt{2}]$ for every $\delta_i\in \Delta$. The value $r_D(\delta_i)$ represents the radius of the circle centered at the point $(\mu_D(\delta_i),\nu_D(\delta_i))$ on the plane. This circle visually represents the membership degree and non-membership degree associated with $\delta_i\in \Delta$.   
\end{definition}

\begin{definition}
Consider two values $\mu_\theta$ and $\nu_\theta$  within the range $[0,\sqrt{2}]$, satisfying the condition $\mu_\theta+\nu_\theta \leq 1$, and a value $r_\theta$ within the range $[0,1]$. The triple $\theta = \langle (\mu_\theta, \nu_\theta); r_\theta\rangle$ is called both C-IFV and D-IFV. That is, C-IFV and D-IFV coincide.    
\end{definition}

\begin{remark} 
Every C-IFS $C_r$ can be expressed as a special case of a D-IFS by assigning a constant value $r$ to the $\hat{r}_D$ radius function. In other words, we can write any C-IFS $C_r$ as:
\[
C_r = D_{\hat{r}=r}= \{\langle \delta_i, (\mu_D(\delta_i), \nu_D(\delta_i)); r\rangle: \delta_i \in \Delta\}.
\]
This implies that every C-IFS is also a D-IFS. Each D-IFS may not necessarily be a C-IFS. Hence, D-IFS is considered as an extension of the C-IFS framework.   
\end{remark}

The following definition outlines the set operations defined for D-IFSs. 

\begin{definition}  \label{subset}
    Consider two D-IFSs, $D_{\hat{r}}$ and $E_{\hat{r}}$ and they are defined as follows:
\[
D_{\hat{r}} = \{\langle \delta_i, (\mu_D(\delta_i,), \nu_D(\delta_i,)); {\hat{r}}_D(\delta_i) \rangle : \delta_i \in \Delta\} \mbox{ and } E_{\hat{r}} = \{\langle \delta_i, (\mu_E(\delta_i), \nu_E(\delta_i)); {\hat{r}}_E(\delta_i) \rangle : \delta_i \in \Delta\}.
\]
Some set operations between two D-IFSs can be defined as follows:

\textbf{i)} $D_{\hat{r}} \subset E_{\hat{r}}$ if and only if ${\hat{r}_D}(\delta_i) \leq {\hat{r}_E}(\delta_i)$, $\mu_D(\delta_i) \leq \mu_E(\delta_i)$, and $\nu_D(\delta_i) \geq \nu_E(\delta_i)$ for each $\delta_i \in \Delta$.

\textbf{ii)} $D_{\hat{r}} = E_{\hat{r}}$ if and only if $D_{\hat{r}} \subset E_{\hat{r}}$ and $E_{\hat{r}} \subset D_{\hat{r}}$.

\textbf{iii)}  The complement of $D_{\hat{r}}$ is denoted as $D^c_{\hat{r}}$ and is defined as $D^c_{\hat{r}} = \{\langle \delta_i, (\nu_D(\delta_i), \mu_D(x_i)); {\hat{r}}_D(\delta_i) \rangle : \delta_i \in \Delta\}.$

\end{definition}

We propose an approach to ranking two DIFVs using both the score function and the accuracy function. These functions are outlined as follows.

\begin{definition}\label{scoref}
    Let $\theta=\langle (\mu_\theta,\nu_\theta);r_\theta\rangle$ represent a D-IFV (or C-IFV) and let $\xi \in [0,1]$. A score function $\mathcal{S}$ and the accuracy function $\mathcal{H}$ are defined for D-IFVs as 
    \[
\mathcal{S}(\theta)=\xi \Bigg( \frac{\mu_\theta-\nu_\theta+1}{2}\Bigg)+(1-\xi)\frac{r_\theta}{\sqrt{2}}
    \]
    and 
    \[
    \mathcal{H}(\theta)=\xi( \mu_\theta+\nu_\theta)+(1-\xi)\frac{r_\theta}{\sqrt{2}},
    \]
respectively. One can see that $0\leq\mathcal{S}(\theta)\leq 1$ and $0\leq\mathcal{H}(\theta)\leq 1$. By altering  $\xi\in [0,1]$, a decision maker has the ability to control how much intuitionistic fuzzy information and radius information influence the score function's outcome. When $\xi$ equals $1$, the the degree of radius is effectively excluded from the decision-making process.
\end{definition}

\begin{definition}
Consider two D-IFVs, $\theta=\langle (\mu_\theta,\nu_\theta);r_\theta\rangle$ and $\kappa=\langle (\mu_\kappa,\nu_\kappa);r_\kappa\rangle$. $\theta$ and $\kappa$  can be ranked by the following procedure:
\begin{itemize}
\item \textbf{(1)} If $\mathcal{S}(\theta)<\mathcal{S}(\kappa)$, then $\theta$ is ranked lower than $\kappa$, showed by $\theta \preceq \kappa$.
\item \textbf{(2)} If $\mathcal{S}(\theta)=\mathcal{S}(\kappa)$,
\begin{itemize}
\item \textbf{(2a)} If $\mathcal{H}(\theta)<\mathcal{H}(\kappa)$, then $\theta$ is ranked lower than $\kappa$.
\item \textbf{(2b)} If $\mathcal{H}(\theta)=\mathcal{H}(\theta)$, then $\theta$ is equivalent to $\kappa$, showed by $\theta \equiv \kappa$.
\end{itemize}
\end{itemize}
\end{definition}

Some algebraic operations defined for C-IFVs are also valid for D-IFVs \cite{Garg}. We extend the existing algebraic operations, which are also valid for C-IFVs, by considering that the radius lies between $0$ and $\sqrt{2}$.

\begin{definition}\label{algebraic}
   Let $\theta=\langle (\mu_\theta,\nu_\theta);r_\theta\rangle$ and $\kappa=\langle (\mu_\kappa,\nu_\kappa);r_\kappa\rangle$ represent two D-IFVs (or C-IFVs) and $\zeta>0$. We have some algebraic operations based on radius between D-IFVs as follows:

   \textbf{a)} $\theta \oplus_{q} \kappa=\displaystyle{\Big\langle \Big(1-(1-\mu_\theta)(1-\mu_\kappa),\nu_\theta\nu_\kappa\Big); \frac{r_\theta r_\kappa}{\sqrt{2}}}\Big\rangle$ 

\textbf{b)} $\theta \oplus_{p} \kappa=\displaystyle{\Big\langle \Big(1-(1-\mu_\theta)(1-\mu_\kappa),\nu_\phi\nu_\kappa\Big);\sqrt{2}-\sqrt{2}\Big(1-\frac{r_\theta}{\sqrt{2}}\Big)\Big(1-\frac{r_\kappa}{\sqrt{2}}\Big)}\Big\rangle$ 

\textbf{c)} $\theta \otimes_{q} \kappa=\displaystyle{\Big\langle \big(\mu_\theta\mu_\kappa,1-(1-\nu_\theta)(1-\nu_\kappa)\big);\frac{r_\theta r_\kappa}{\sqrt{2}}}\Big\rangle$ 

\textbf{d)} $\theta \otimes_{p} \kappa=\displaystyle{\Big\langle \big(\mu_\theta\mu_\kappa,1-(1-\nu_\theta)(1-\nu_\kappa)\big);\sqrt{2}-\sqrt{2}\Big(1-\frac{r_\theta}{\sqrt{2}}\Big)\Big(1-\frac{r_\kappa}{\sqrt{2}}\Big)} \Big\rangle$

\textbf{e)} $\zeta_q\theta=\displaystyle{\Big\langle( 1-(1-\mu_\theta)^\zeta,\nu_\theta^\zeta);\sqrt{2}\Big(\frac{r_\theta}{\sqrt{2}}\Big)^\zeta}\Big\rangle$

\textbf{f)} $\zeta_p\theta=\displaystyle{\Big\langle (1-(1-\mu_\theta)^\zeta,\nu_\theta^\zeta);\sqrt{2}-\sqrt{2}\Big(1-\frac{r_\theta}{\sqrt{2}}\Big)^\zeta}\Big\rangle$

\textbf{g)} $\theta^{\zeta_q}=\displaystyle{\Big\langle (\mu_\theta^\zeta,1-(1-\nu_\theta)^\zeta);\sqrt{2}\Big(\frac{r_\theta}{\sqrt{2}}\Big)^\zeta}\Big\rangle$

\textbf{h)} $\theta^{\zeta_p}=\displaystyle{\Big\langle (\mu_\theta^\zeta,1-(1-\nu_\theta)^\zeta);\sqrt{2}-\sqrt{2}\Big(1-\frac{r_\theta}{\sqrt{2}}\Big)^\zeta}\Big\rangle$
 \end{definition}

\begin{definition} \label{dıfwa} 
    Consider a collection of D-IFVs denoted as $\{\theta_k=\langle (\mu_{\theta_k},\nu_{\phi_k});r_{\theta_k}\rangle:k=1,\ldots,m\}$.  Weighted arithmetic aggregation operators are mappings defined as
    \[
\mbox{D-IFWAO}_q(\theta_1,...,\theta_m):=(q)\bigoplus_{k=1}^m {\omega_{k_q}}\theta_k.
\]
and 
\[
\mbox{D-IFWAO}_p(\theta_1,...,\theta_m):=(p)\bigoplus_{k=1}^m {\omega_{k_p}}\theta_k.
\]
respectively. Here, each weight $\omega_i$ satisfies $0\leq \omega_i\leq 1$ for $k=1,\ldots,m$, and the sum of all weights equals 1.
\end{definition}
It can be easily obtained from Definition \ref{dıfwa}  that
   \begin{eqnarray*}
\mbox{\textit{D-IFWAO}}_q(\theta_1,\ldots,\theta_m)&=&\Bigg\langle \Bigg( 1-\prod_{k=1}^m(1-\mu_{\theta_{k}})^{\omega_k}, \prod_{k=1}^m\nu_{\theta_k}^{\omega_k}\Bigg); \prod_{k=1}^m r_{\theta_k}^{\omega_k} \Bigg\rangle,
\end{eqnarray*}
and
\begin{eqnarray*}
\mbox{\textit{D-IFWAO}}_p(\theta_1,\ldots,\theta_m)&=&\Bigg\langle \Bigg(1-\prod_{k=1}^m(1-\mu_{\theta_{k}})^{\omega_k}, \prod_{k=1}^m \nu_{\theta_k}^{\omega_k}\Bigg); \sqrt{2}-\prod_{k=1}^m\Big(\sqrt{2}-r_{\theta_{k}}\Big)^{\omega_k} \Bigg\rangle,
\end{eqnarray*}
respectively. 

\begin{definition}\label{dıfwg}  Given a collection of D-IFVs denoted as $\{\theta_k=\langle (\mu_{\theta_k},\nu_{\theta_k});r_{\theta_k}\rangle:k=1,\ldots,m\}$.
 Weighted geometric aggregation operators are mappings defined as
\[
\mbox{D-IFWGO}_q(\theta_1,\ldots,\theta_m):=(q)\bigotimes_{k=1}^m\theta_k^{\omega_{k_q}}.
\]
and
\[
\mbox{D-IFWGO}_p(\theta_1,...,\theta_m):=(p)\bigotimes_{k=1}^m \theta_k^{\omega_{k_p}}.
\]
respectively. Here, $0\leq \omega_k\leq 1$ for any $k=1,\ldots,m$, and the weights satisfy the condition $\sum_{k=1}^m w_k=1$.   
\end{definition}

It can be seen from Definition \ref{dıfwg} that
\begin{eqnarray*}
\mbox{\textit{D-IFWGO}}_q(\theta_1,\ldots,\theta_m)&=&\Bigg\langle \Bigg( \prod_{k=1}^m \mu_{\theta_k}^{\omega_k}, 1-\prod_{k=1}^m(1-\nu_{\theta_{k}})^{\omega_k} \Bigg) ; \prod_{k=1}^m r_{\theta_k}^{\omega_k} \Bigg\rangle,
\end{eqnarray*}
and
\begin{eqnarray*}
\mbox{\textit{D-IFWGO}}_p(\theta_1,\ldots,\theta_m)&=&\Bigg\langle \Bigg( \prod_{k=1}^m  \mu_{\theta_k}^{\omega_k}, 1-\prod_{i=1}^m(1-\nu_{\theta_{k}})^{\omega_k} \Bigg) ; \sqrt{2}-\prod_{k=1}^m\Big(\sqrt{2}-r_{\theta_{k}}\Big)^{\omega_k} \Bigg\rangle,
\end{eqnarray*}
respectively.

\section{Disc Intuitionistic fuzzy Choquet integral operator} \label{choquett}

This section begins with a review of the fuzzy measure and the Choquet integral, followed by the introduction of Choquet integral operators specifically designed for D-IFVs

\begin{definition} \cite{Sugeno} Let $X$ be a universal set, and let $2^X$ represent the power set of $X$. A fuzzy measure is defined as a set function \(\tau: 2^X \to [0,1]\) satisfying the following properties: 

\textbf{1)} \(\tau(\emptyset) = 0\) and \(\tau(X) = 1\). 

\textbf{2)} For any two subsets \(A, B \subseteq X\), if \(A \subseteq B\), then \(\tau(A) \leq \tau(B)\). \end{definition}

Sugeno \cite{Sugeno} proposed the concept of the $\lambda$-fuzzy measure, which incorporates an interaction index and criterion weighting. This measure is particularly useful for capturing synergies between criteria and offers more accurate solutions in complex decision-making contexts.

\begin{definition} \cite{Sugeno} Let \(X\) be a set, and consider a set function \(\tau: 2^X \to [0,1]\) with \(\lambda \geq -1\). The function \(\tau\) is called a \(\lambda\)-fuzzy measure if it satisfies the following condition for any disjoint subsets \(A\) and \(B\) of \(X\): 
\[
\tau(A \cup B) = \tau(A) + \tau(B) + \lambda \tau(A)\tau(B).
\]
\end{definition}

\begin{definition} 
Let $X$ be a finite set and let $\tau:2^X\to [0,1]$ be a fuzzy measure. The discrete Choquet integral a function $f: X \to [0,1]$ with respect to $\tau$ is defined by
\[
C_\tau(f(x_1),\ldots, f(x_m)):=\sum_{k=1}^m \big[\tau(E_{(k)})-\tau(E_{(k+1)})\big]f(x_{(k)}),
\]
where the sequence $\{x_{(k)}\}_{k=0}^m$ indicates the indices permuted such that $0:=f(x_{(0)})\leq f(x_{(1)})\leq \ldots \leq f(x_{(m)})$, $E_{(k)}:=\{x_{(k)},x_{(k+1)},\ldots,x_{(m)}\}$ and $E_{(m+1)}=\emptyset$.
\end{definition}
We now define some arithmetic aggregation operators.
\begin{definition}
     Consider a collection of D-IFVs denoted as $\{\theta_k=\langle (\mu_{\theta_k},\nu_{\phi_k});r_{\theta_k}\rangle:k=1,\ldots,m\}$.  Disc intuitionistic fuzzy  Choquet arithmetic integral operators with respect to $\tau$ are defined by

     \[
     \mbox{D-IFCAIO}_q^\tau(\theta_1,\ldots,\theta_m)= (q)\bigoplus_{k=1}^m  \big[\tau(F_{(k)})-\tau(F_{(k+1)})\big]_q\theta_{(k)},
     \]
and
\[
      \mbox{D-IFCAIO}_p^\tau(\theta_1,\ldots,\theta_m)= (p)\bigoplus_{k=1}^m  \big[\tau(F_{(k)})-\tau(F_{(k+1)})\big]_p\theta_{(k)},
     \]   
  where the sequence $\{\theta_{(k)}\}_{k=0}^m$ indicates the indices permuted such that $\theta_{(1)}\preceq\theta_{(2)}\preceq\ldots \preceq \theta_{(m)}$, $F_{(k)}:=\{(k),(k+1),\ldots,(m)\}$ and $F_{(m+1)}=\emptyset$.   
\end{definition}
The following theorem formulates the given arithmetic aggregation operators.
\begin{theorem}\label{adifc}
      Consider a collection of D-IFVs denoted as $\{\theta_k=\langle (\mu_{\theta_k},\nu_{\phi_k});r_{\theta_k}\rangle:k=1,\ldots,m\}$. $\mbox{D-IFCAIO}_q^\tau(\theta_1,\ldots,\theta_m)$ and $\mbox{D-IFCAIO}_p^\tau(\theta_1,\ldots,\theta_m)$ are also a D-IFV and can be expressed by

     \begin{eqnarray*}   
    \mbox{D-IFCAIO}_q^\tau(\theta_1,\ldots,\theta_m)&=&\Bigg\langle \Bigg( 1-\prod_{k=1}^m(1-\mu_{\theta_{(k)}})^{\tau(F_{(k)})-\tau(F_{(k+1)})}, \prod_{k=1}^m\nu_{\theta_{(k)}}^{\tau(F_{(k)})-\tau(F_{(k+1)})}\Bigg);{}\\
&& {} \prod_{k=1}^m r_{\theta_{(k)}}^{\tau(F_{(k)})-\tau(F_{(k+1)})} \Bigg\rangle
   \end{eqnarray*} 
and 

\begin{eqnarray*}
\mbox{D-IFCAIO}_p^\tau(\theta_1,\ldots,\theta_m)&=& \Bigg\langle \Bigg(1-\prod_{k=1}^m(1-\mu_{\theta_{(k)}})^{\tau(F_{(k)})-\tau(F_{(k+1)})}, \prod_{k=1}^m \nu_{\theta_{(k)}}^{\tau(F_{(k)})-\tau(F_{(k+1)})}\Bigg);{}\\
&& {} \sqrt{2}-\prod_{k=1}^m\Big(\sqrt{2}-r_{\theta_{(k)}}\Big)^{\tau(F_{(k)})-\tau(F_{(k+1)})} \Bigg\rangle
\end{eqnarray*}
respectively where the sequence $\{\theta_{(k)}\}_{k=0}^m$ indicates the indices permuted such that $\theta_{(1)}\preceq\theta_{(2)}\preceq\ldots \preceq \theta_{(m)}$, $F_{(k)}:=\{(k),(k+1),\ldots,(m)\}$ and $F_{(m+1)}=\emptyset$.     
\end{theorem}

\begin{proof}
   It is clear that $\mbox{\textit{D-IFCAIO}}_q^\tau(\theta_1,\ldots,\theta_m)$ constitutes a D-IFV. The validity of the second part becomes apparent through the application of mathematical induction. When $k=2$, with the help of Definition \ref{algebraic} we have

   \[
   \Big[\tau(F_{(1)})-\tau(F_{(2)})\big]_q\theta_{(1)}=\Bigg\langle\Big(1-(1-\mu_{\theta_{(1)}})^{\tau(F_{(1)})-\tau(F_{(2)})},\nu_{\theta_{(1)}}^{\tau(F_{(1)})-\tau(F_{(2)})}\Big);\sqrt{2}\Big(\frac{r_{\theta_{(1)}}}{\sqrt{2}}\Big)^{\tau(F_{(1)})-\tau(F_{(2)})}\Bigg\rangle,
   \]
and
 \[
   \Big[\tau(F_{(2)})-\tau(F_{(3)})\big]_q\theta_{(2)}=\Bigg\langle\Big(1-(1-\mu_{\theta_{(2)}})^{\tau(F_{(2)})-\tau(F_{(3)})},\nu_{\theta_{(2)}}^{\tau(F_{(2)})-\tau(F_{(3)})}\Big);\sqrt{2}\Big(\frac{r_{\theta_{(2)}}}{\sqrt{2}}\Big)^{\tau(F_{(2)})-\tau(F_{(3)})}\Bigg\rangle.
   \]
Since $F_{(3)}=\emptyset$, we obtain
\begin{eqnarray*}
    \mbox{\textit{D-IFCAIO}}_q^\tau(\theta_1,\theta_2)&=& \Big[\tau(F_{(2)})-\tau(F_{(3)})\big]_p\theta_{(2)} \oplus_p \Big[\tau(F_{(2)})-\tau(F_{(3)})\big]_p\theta_{(2)}\\
    &=& \Bigg\langle\Big(1-(1-\mu_{\theta_{(1)}})^{\tau(F_{(1)})-\tau(F_{(2)})}(1-\mu_{\theta_{(2)}})^{\tau(F_{(2)})-\tau(F_{(3)})},{}\\
&& {}\nu_{\theta_{(1)}}^{\tau(F_{(1)})-\tau(F_{(2)})}\nu_{\theta_{(2)}}^{\tau(F_{(2)})-\tau(F_{(3)})}\Big);\sqrt{2}\Big(\frac{r_{\theta_{(1)}}}{\sqrt{2}}\Big)^{\tau(F_{(1)})-\tau(F_{(2)})}\Big(\frac{r_{\theta_{(2)}}}{\sqrt{2}}\Big)^{\tau(F_{(2)})-\tau(F_{(3)})}\Bigg\rangle\\
&=& \Bigg\langle\Big(1-(1-\mu_{\theta_{(1)}})^{\tau(F_{(1)})-\tau(F_{(2)})}(1-\mu_{\theta_{(2)}})^{\tau(F_{(2)})-\tau(F_{(3)})},{}\\
&& {}\nu_{\theta_{(1)}}^{\tau(F_{(1)})-\tau(F_{(2)})}\nu_{\theta_{(2)}}^{\tau(F_{(2)})-\tau(F_{(3)})}\Big);r_{\theta_{(1)}}^{\tau(F_{(1)})-\tau(F_{(2)})}r_{\theta_{(2)}}^{\tau(F_{(2)})-\tau(F_{(3)})}\Bigg\rangle
\end{eqnarray*}
Now suppose that the expression
\begin{eqnarray*}
\mbox{\textit{D-IFCAIO}}_q^\tau(\theta_1,\ldots,\theta_n)&=&\Bigg\langle \Bigg( 1-\prod_{k=1}^n(1-\mu_{\theta_{(k)}})^{\tau(F_{(k)})-\tau(F_{(k+1)})}, \prod_{k=1}^n\nu_{\theta_{(k)}}^{\tau(F_{(k)})-\tau(F_{(k+1)})}\Bigg);{}\\
&& {} \prod_{k=1}^n r_{\theta_{(k)}}^{\tau(F_{(k)})-\tau(F_{(k+1)})} \Bigg\rangle
\end{eqnarray*}
is true for $m=n$. Then for $m=n+1$ it is attain that
\begin{eqnarray*}
 \mbox{\textit{D-IFCAIO}}_q^\tau(\theta_1,\ldots,\theta_n,\theta_{n+1})&=& \mbox{\textit{D-IFCAIO}}_p^\tau(\theta_1,\ldots,\theta_n) \oplus_p \Big[\tau(F_{(n+1)})-\tau(F_{(n+2)})\big]_p\theta_{(n+1)}\\&=&\Bigg\langle \Bigg( 1-\prod_{k=1}^{n+1}(1-\mu_{\theta_{k}})^{\tau(F_{(k)})-\tau(F_{(k+1)})}, \prod_{k=1}^{n+1}\nu_{\theta_k}^{\tau(F_{(k)})-\tau(F_{(k+1)})}\Bigg); {}\\
&& {}\prod_{k=1}^{n+1} r_{\theta_{(k)}}^{\tau(F_{(k)})-\tau(F_{(k+1)})} \Bigg\rangle
\end{eqnarray*}
  Therefore, the proof is finished for $\mbox{\textit{D-IFCAIO}}_q^\tau(\theta_1,\ldots,\theta_m)$. In similar way, one can easily proof for $\mbox{\textit{D-IFCAIO}}_p^\tau(\theta_1,\ldots,\theta_m)$.
\end{proof}
In the following we define some geometric aggregation operators.
\begin{definition}
     Consider a collection of D-IFVs denoted as $\{\theta_k=\langle (\mu_{\theta_k},\nu_{\phi_k});r_{\theta_k}\rangle:k=1,\ldots,m\}$. Disc intuitionistic fuzzy Choquet geometric integral operators with respect to $\tau$ are defined by

     \[
     \mbox{\textit{D-IFCGIO}}_q^\tau(\theta_1,\ldots,\theta_m)= (q)\bigotimes_{k=1}^m  \theta_{(k)}^{\tau(F_{(k)})-\tau(F_{(k+1)})},
     \]
and
\[
     \mbox{\textit{D-IFCGIO}}_p^\tau(\theta_1,\ldots,\theta_m)= (p)\bigotimes_{k=1}^m  \theta_{(k)}^{\tau(F_{(k)})-\tau(F_{(k+1)})},
     \]   
  where the sequence $\{\theta_{(k)}\}_{k=0}^m$ indicates the indices permuted such that $\theta_{(1)}\preceq\theta_{(2)}\preceq\ldots \preceq \theta_{(m)}$, $F_{(k)}:=\{(k),(k+1),\ldots,(m)\}$ and $F_{(m+1)}=\emptyset$.   
\end{definition}
The following theorem formulates the geometric aggregation operators.
\begin{theorem}
      Consider a collection of D-IFVs denoted as $\{\theta_k=\langle (\mu_{\theta_k},\nu_{\theta_k});r_{\theta_k}\rangle:k=1,\ldots,m\}$. $\mbox{\textit{D-IFCGIO}}_q^\tau(\theta_1,\ldots,\theta_m)$ and $\mbox{\textit{D-IFCGIO}}_p^\tau(\theta_1,\ldots,\theta_m)$ are also D-IFV and can be expressed by

      \begin{eqnarray*}
    \mbox{\textit{D-IFCGIO}}_q^\tau(\theta_1,\ldots,\theta_m)&=&\Bigg\langle \Bigg( \prod_{k=1}^m \mu_{\theta_{(k)}}^{\tau(F_{(k)})-\tau(F_{(k+1)})}, 1-\prod_{k=1}^m(1-\nu_{\theta_{(k)}})^{\tau(F_{(k)})-\tau(F_{(k+1)})} \Bigg) ;{}\\
&& {}  \prod_{k=1}^n r_{\theta_{(k)}}^{\tau(F_{(k)})-\tau(F_{(k+1)})} \Bigg\rangle
    \end{eqnarray*}
and 

\begin{eqnarray*}
\mbox{\textit{D-IFCGIO}}_p^\tau(\theta_1,\ldots,\theta_m)&=& \Bigg\langle \Bigg( \prod_{k=1}^m \mu_{\theta_{(k)}}^{\tau(F_{(k)})-\tau(F_{(k+1)})}, 1-\prod_{k=1}^m(1-\nu_{\theta_{(k)}})^{\tau(F_{(k)})-\tau(F_{(k+1)})} \Bigg);{}\\
&& {} \sqrt{2}-\prod_{k=1}^m\Big(\sqrt{2}-r_{\theta_{(k)}}\Big)^{\tau(F_{(k)})-\tau(F_{(k+1)})} \Bigg\rangle
\end{eqnarray*}
respectively where the sequence $\{\theta_{(k)}\}_{k=0}^m$ indicates the indices permuted such that $\theta_{(1)}\preceq\theta_{(2)}\preceq\ldots \preceq \theta_{(m)}$, $F_{(k)}:=\{(k),(k+1),\ldots,(m)\}$ and $F_{(m+1)}=\emptyset$.     
\end{theorem}
\begin{proof}
    It can be proved similar to Theorem \ref{adifc}.
\end{proof}

\begin{remark}
    Let $\theta_1=\langle (\mu_{\theta_1},\nu_{\theta_1});r_{\theta_1}\rangle$ and $\theta_2=\langle (\mu_{\theta_2},\nu_{\theta_2});r_{\theta_2}\rangle$ such that $\mu_{\theta_1}\leq \mu_{\theta_2}, \nu_{\theta_1}\geq \nu_{\theta_2},r_{\theta_1}\leq r_{\theta_2}$ or $\mu_{\theta_1}\geq \mu_{\theta_2}, \nu_{\theta_1}\leq \nu_{\theta_2},r_{\theta_1}\geq r_{\theta_2}$. Assume that the weighted vector is $\omega=(\omega_1,\omega_2)=(0.3,0.7)$ and a fuzzy measure $\tau$ is constructed by $\tau(\varnothing)=0, \tau(\{1\})=0.22, \tau(\{2\})=0.58,$ and $\tau(\{1,2\})=1$. Then we have

\begin{eqnarray*}
    \mbox{\textit{D-IFWGO}}_p(\theta_1,\theta_2)&=& \Big\langle \Big(\mu_{\mbox{\textit{D-IFWGO}}_p},\nu_{\mbox{\textit{D-IFWGO}}_p}\Big);r_{\mbox{\textit{D-IFWGO}}_p}  \Big\rangle  \\
    &=&\Big\langle \Big(\mu_{\theta_1}^{0.3}\mu_{\theta_2}^{0.7},1-(1-\nu_{\theta_1})^{0.3}(1-\nu_{\theta_2})^{0.7}\Big);\sqrt{2}-(\sqrt{2}-\mu_{\theta_1})^{0.3}(\sqrt{2}-\mu_{\theta_2})^{0.7}\Big\rangle
\end{eqnarray*}
and
\begin{eqnarray*}
    \mbox{\textit{D-IFCGIO}}_p^\tau(\theta_1,\theta_2)&=& \Big\langle \Big(\mu_{\mbox{\textit{D-IFCGIO}}_p^\tau},\nu_{\mbox{\textit{D-IFCGIO}}_p^\tau}\Big);r_{\mbox{\textit{D-IFCGIO}}_p^\tau}  \Big\rangle   
\end{eqnarray*}
where 
$$ \mu_{\mbox{\textit{D-IFCGIO}}_p^\tau}=
\begin{cases} 
     \mu_{\theta_1}^{0.42}\mu_{\theta_2}^{0.58}&, \mbox{ if } \mu_{\theta_1}\leq \mu_{\theta_2} \\
      \mu_{\theta_1}^{0.78}\mu_{\theta_2}^{0.22} &, \mbox{ if } \mu_{\theta_1}\geq \mu_{\theta_2}\\
\end{cases}$$
$$ \nu_{\mbox{\textit{D-IFCGIO}}_p^\tau}=
\begin{cases} 
     1-(1-\nu_{\theta_1})^{0.42}(1-\nu_{\theta_2})^{0.58}&, \mbox{ if } \nu_{\theta_1}\geq \nu_{\theta_2} \\
      1-(1-\nu_{\theta_1})^{0.78}(1-\nu_{\theta_2})^{0.22} &, \mbox{ if } \nu_{\theta_1}\leq \nu_{\theta_2}\\
\end{cases}$$
$$ r_{\mbox{\textit{D-IFCGIO}}_p^\tau}=
\begin{cases} 
     \sqrt{2}-(\sqrt{2}-r_{\theta_1})^{0.42}(\sqrt{2}-r_{\theta_2})^{0.58}&, \mbox{ if } r_{\theta_1}\leq r_{\theta_2} \\
      \sqrt{2}-(\sqrt{2}-r_{\theta_1})^{0.78}(\sqrt{2}-r_{\theta_2})^{0.22} &, \mbox{ if } r_{\theta_1}\geq r_{\theta_2}\\
\end{cases}$$
    Figure \ref{chqouetweigtedcom} shows a comparison of operators $\mbox{\textit{D-IFWGO}}_p$ and $\mbox{\textit{D-IFCGIO}}_p^\tau$. It can be seen that these two operators have different characteristics.
\end{remark}

\begin{figure}[h!]
    \centering
    \includegraphics[scale=0.55]{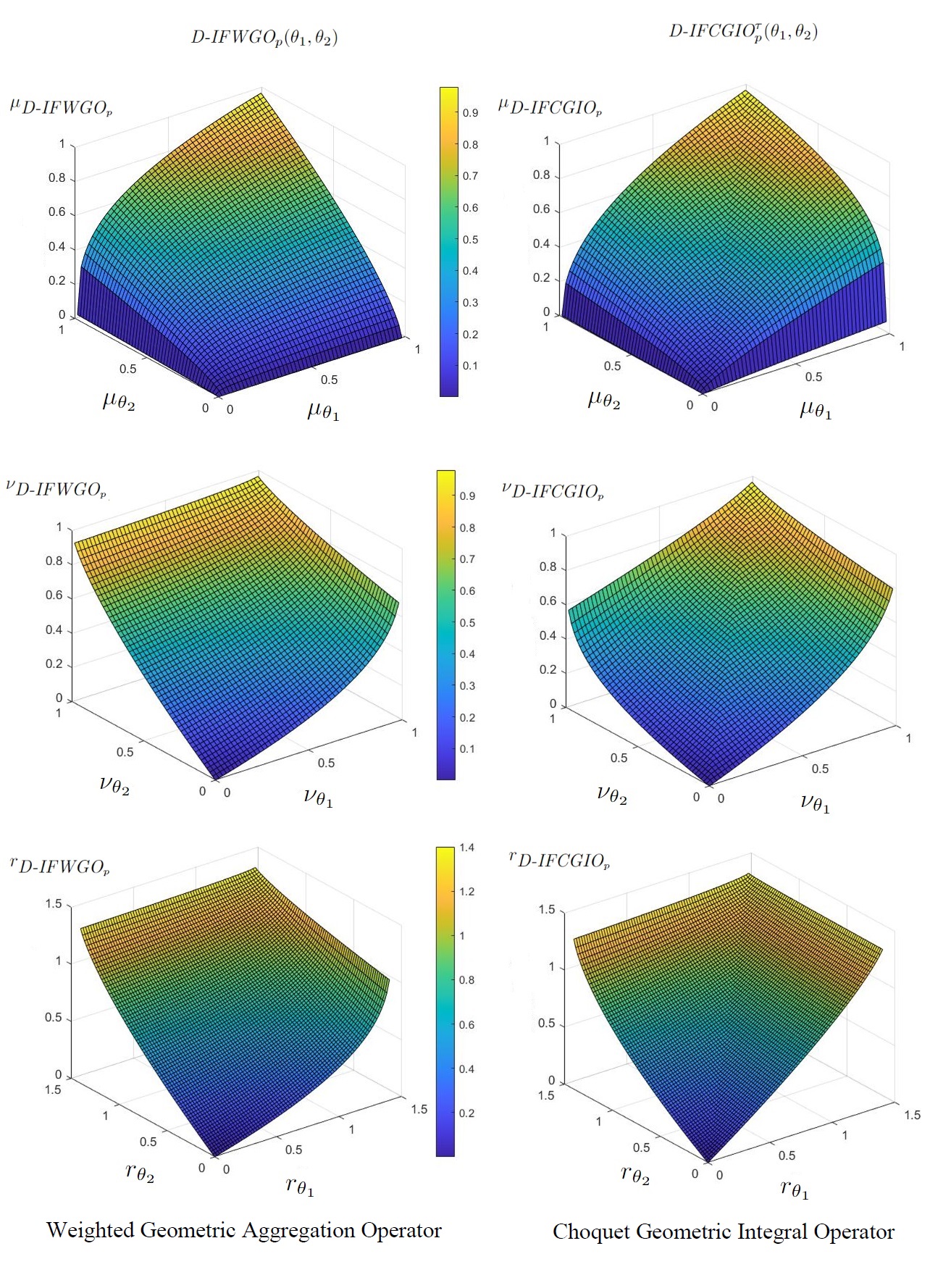}
    \caption{A comparison of Choquet geometric integral operator (right side) and weighted geometric aggregation operator (left side)}
    \label{chqouetweigtedcom}
\end{figure}

\begin{theorem}
    Consider a collection of D-IFVs denoted as $\{\theta_k=\langle (\mu_{\theta_k},\nu_{\phi_k});r_{\theta_k}\rangle:k=1,\ldots,m\}$ and $\omega=(\omega_1,\ldots,\omega_m)$ is a weight vector. If $\tau$ is an additive measure such that $\tau([m])=\sum_{k=1}^m \omega_k$, then 
\[
\mbox{\textit{D-IFCAIO}}_q^\tau(\theta_1,\ldots,\theta_m)=\mbox{\textit{D-IFWAO}}_q(\theta_1,\ldots,\theta_m),
\]
 \[
\mbox{\textit{D-IFCAIO}}_p^\tau(\theta_1,\ldots,\theta_m)=\mbox{\textit{D-IFWAO}}_p(\theta_1,\ldots,\theta_m),
\] 
\[
\mbox{\textit{D-IFCGIO}}_q^\tau(\theta_1,\ldots,\theta_m)=\mbox{\textit{D-IFWGO}}_q(\theta_1,\ldots,\theta_m),
\]
and 
\[
\mbox{\textit{D-IFCGIO}}_p^\tau(\theta_1,\ldots,\theta_m)=\mbox{\textit{D-IFWGO}}_p(\theta_1,\ldots,\theta_m).
\]
\end{theorem}
\begin{proof}
    The proof can be easily seen.
\end{proof}

\section{An Application to Solar Panels} \label{appsolar}

This section introduces the CASPAS method following an analysis of the solar panel decision problem. The proposed method is then applied to a specific decision-making scenario.

\subsection{An assessment of the types of solar panels}
%As non-renewable energy sources become scarcer and their adverse environmental impacts more apparent, the significance of renewable energy sources has steadily risen. This shift is driven by the need to provide abundant and environmentally friendly sources of clean energy. A contemporary source of renewable energy is the solar panels, known for its minimal environmental impact and remarkable productivity \cite{Sathaye, Sorensen}. A solar panel, also known as a photovoltaic  panel, is a device that converts light from the sun, which is composed of particles of energy called photons, into electricity that can be used to power electrical loads. Solar panels are used in a wide variety of applications including remote power systems for cabins, telecommunications equipment, remote sensing, and of course for the production of electricity by residential and commercial solar electric systems. A solar panel, comprising a network of linked solar cells, serves as an energy-generating apparatus that transforms solar energy into electrical power through the photovoltaic effect. This conversion process was initially discovered by Becquerel. Solar panels offer various benefits, harnessing renewable and environmentally friendly energy, diminishing greenhouse gas emissions, and cutting down on electricity expenses. However, they come with drawbacks, relying on sunlight availability, demanding periodic cleaning, and involving substantial initial expenses. 

Various kinds of solar panels are available. Different types and technologies have been developed according to various usage needs and conditions. Main types and characteristics of solar panels are summarized Table \ref{tab:my_label3333}.

\footnotesize

\begin{table}
    \centering
\begin{tabular}{l|m{9cm}}
\hline
  \textbf{Solar Panels}   & \textbf{Characteristics} \\ \hline
  Monocrystalline 
  
  Silicon Panel ($P_1$) &  \begin{itemize}
    \item Made using single-crystal silicon.
    \item Typically high efficiency (around 18-22\%).
    \item Black cells with a characteristic look, often with cut corners.
    \item Generally more expensive than other types.
    \item  Suitable for smaller spaces where high energy output is required \cite{Dobrzański, Parida}.
\end{itemize}\\ \hline

Polycrystalline Solar Panel ($P_2$) & \begin{itemize}
    \item Produced using multi-crystal silicon.
    \item Usually slightly lower than monocrystalline panels (around 15-17\%).
    \item Blue-hued cells with straight edges.
    \item More economical, generally lower priced.
    \item  Preferred for large areas where cost-effectiveness is a priority \cite{Dobrzański, Parida}.
\end{itemize} \\ \hline

Thin-Film Solar Panel ($P_3$) & \begin{itemize}
    \item Made from various materials (e.g., amorphous silicon, CdTe, CIGS).
    \item Generally lower efficiency (around 10-13\%).
    \item Flexible structures, often with a darker color.
    \item Cheaper to manufacture but may produce less power per square meter.
    \item  Suitable for large areas; flexible nature advantageous for certain specialized applications \cite{Lee}. 
\end{itemize}\\ \hline
 PERC Solar Panel ($P_4$) & \begin{itemize}
    \item  Additional technologies added to traditional monocrystalline or polycrystalline cells.
    \item Higher efficiency compared to standard panels (above 20\%).
    \item Similar to monocrystalline or polycrystalline panels.
    \item The advanced manufacturing process can make PERC panels slightly more expensive than traditional panels, but the cost is often offset by their higher efficiency.
    \item Used for achieving high energy efficiency in limited spaces \cite{Green}.
\end{itemize} \\ \hline
Bifacial Solar Panel ($P_5$) & \begin{itemize}
    \item Cells can absorb light from both front and back surfaces. Bifacial panels are designed with solar cells that are capable of capturing sunlight on both sides.
    \item Capable of higher energy production due to dual-sided light collection (often generating 5-30\% more energy compared to standard panels under optimal conditions).
    \item More effective when mounted over reflective surfaces.
    \item Higher upfront costs compared to traditional panels due to their advanced technology and materials.
    \item Suitable for a variety of installations, including commercial, residential, and utility-scale solar projects \cite{Guerrero-Lemus}.
\end{itemize}\\

\hline
\end{tabular}
\caption{Types of solar panels and their main characteristics}
    \label{tab:my_label3333}
\end{table}
\normalsize

The choice of solar panel type depends on factors such as efficiency requirements, available space, budget considerations, and the specific application or installation environment. Ongoing research and technological advancements continue to expand the range of options, making solar energy an increasingly versatile and accessible renewable energy source. In the solar panel selection problem to be examined in this article, the criteria to be used in the evaluation according to the opinion of energy systems experts are listed below:
\begin{itemize}
    \item \textbf{Efficiency ($T_1$)}: The efficiency of solar panels refers to their ability to convert sunlight into usable electricity. It is a critical factor in evaluating the performance and economic viability of solar energy systems. Solar panel efficiency is typically expressed as a percentage and represents the ratio of the electrical power output to the solar energy input.

    \item \textbf{Cost ($T_2$)}: 
The cost of solar panels can vary widely depending on several factors, including the type of solar panel, the size of the system, the location of installation, and the specific manufacturer or brand. Prices have been consistently decreasing over the years due to technological advancements and increased market competition. 

 \item \textbf{Durability ($T_3$)}: The durability of solar panels is a measure of their ability to withstand environmental conditions, maintain performance, and continue generating electricity over time without significant degradation. Durability is crucial for the long-term reliability and economic viability of a solar energy system. 

 \item \textbf{Installation Complexity ($T_4$)}: The installation complexity of solar panels refers to the level of difficulty and intricacy involved in the process of setting up a solar panel system. Several factors contribute to the installation complexity, and they can affect the overall cost, time, and expertise required to successfully install solar panels.
\end{itemize}
Choosing the right solar panels involves a careful evaluation of these criteria in relation to the project's specific goals, budget, and location. Consulting with solar energy experts and conducting thorough research on available options can help in making an informed decision.

\subsection{CASPAS approach} \label{waspasch}
 WASPAS was proposed by Zavadskas et al. \cite{Zavadskas1}. WASPAS integrates the principles of both WSM and WPM using a method that applies weights to each. The method involves assigning weights to different criteria to represent their relative importance, normalizing the decision matrix, and then aggregating the scores for each alternative.  WASPAS is known for its flexibility in incorporating both the weighted sum and the weighted product models, allowing decision-makers to adapt the method to different decision contexts. 

In this subsection,  CASPAS is proposed for disc intuitionistic fuzzy setting. This method enables an accurate modelling to be established in the decision-making process by taking into account the interaction between the criteria. We outline the procedural steps for employing the extended D-IF CASPAS method in the examination of MCGDM problems as follows:

\begin{itemize}
\item \textbf{Step I:} Form a MCGDM problem, where there are $m$ alternatives denoted as $P=\{P_1,\ldots,P_m\}$. A group of $n$ experts, identified as $E=\{E_1,\ldots,E_n\}$, assesses these alternatives using $k$ criteria outlined in $T=\{T_1,\ldots,T_k\}$.
\item \textbf{Step II:} Every expert assesses the alternatives by employing the linguistic terms provided in Table \ref{tab:my_label11} which is each of the linguistic terms corresponds to a D-IFV. When there is a cost criterion present, the values for this criterion undergo a complement operation. Thus, a disc intuitionistic fuzzy group normalized decision matrix is created.
 
\begin{table}[h!]
    \centering
    \begin{tabular}{lc}
    \hline
     Linguistic term  & D-IFV  \\ \hline
      Extremely High (EH)  & $\langle ( 0.9,0.1);0.9\rangle$ \\
      Very High (VH)  & $\langle ( 0.8,0.2);0.8\rangle$ \\
       High (H)  & $\langle ( 0.7,0.3);0.7\rangle$ \\
      Moderately High (MH) & $\langle ( 0.6,0.4),0.6\rangle$ \\
      Medium (M)  & $\langle ( 0.5,0.5);0.5\rangle$ \\
      Moderately Low (ML)  & $\langle ( 0.4,0.6);0.4\rangle$ \\
      Low (L)  & $\langle ( 0.3,0.7);0.3\rangle$ \\
      Very Low (VL)  & $\langle ( 0.2,0.8);0.2\rangle$ \\
      Extremely Low (EL)  & $\langle ( 0.1,0.9);0.1\rangle$ \\ \hline
    \end{tabular}
    \caption{Linguistic terms corresponding D-IFV}
    \label{tab:my_label11}
\end{table}

\item\textbf{Step III:} A weight vector of the experts is determined according to the experience of the experts. Then the disc intuitionistic fuzzy group normalized decision matrix  is aggregated by utilizing one of operators $\mbox{\textit{D-IFWAO}}_q^A$, $\mbox{\textit{D-IFWAO}}_p^A$, $\mbox{\textit{D-IFWGO}}_q^A$, and $\mbox{\textit{D-IFWGO}}_p^A$. 

\item \textbf{Step IV:} With the help of Table \ref{tab:my_label11}, the importance weights of the criteria are evaluated by the experts. The importance weights matrix of criteria is formed.

\item\textbf{Step V:} By using weight vector of engineers, the importance weights matrix of criteria  is aggregated similarly to Step III.

\item\textbf{Step VI:} Compute the score function value $\mathcal{S}_j$ of each criterion $T_j$  in  aggregated importance weights matrix of criteria. Then normalize score of each criterion $T_j$ using

\begin{equation}\label{eq1}
    \omega_j=\frac{\mathcal{S}_j}{\sum\limits_{j=1}^k \mathcal{S}_j}
\end{equation}
Thus the weight vector of the criteria is obtained.

\item\textbf{Step VII:} Construct a fuzzy measure $\tau$ to be able compute disc intuitionistic fuzzy
Choquet operators after determining an interaction index $\lambda$ and weight vector of the criteria.

\item\textbf{Step VIII:}  Calculate results of $CSM_i$ based on $\mbox{D-IFCAIO}_q^\tau$ and $\mbox{D-IFCAIO}_p^\tau$ for each alternative $P_i$ using aggregated decision matrix and one of the operators
\begin{equation}\label{eqcsm1}
  CSM_i^q=\Bigg\langle \Bigg( 1-\prod_{j=1}^k(1-\mu_{\theta_{T_{(j)}}})^{\tau(F_{(j)})-\tau(F_{(j+1)})}, \prod_{j=1}^k\nu_{\theta_{T_{(j)}}}^{\tau(F_{(j)})-\tau(F_{(j+1)})}\Bigg); \prod_{j=1}^k r_{\theta_{T_{(j)}}}^{\tau(F_{(j)})-\tau(F_{(j+1)})} \Bigg\rangle  
\end{equation}
or
\begin{equation}\label{eqcsm2}
    CSM_i^p= \Bigg\langle \Bigg(1-\prod_{j=1}^k(1-\mu_{\theta_{T_{(j)}}})^{\tau(F_{(j)})-\tau(F_{(j+1)})}, \prod_{j=1}^k \nu_{\theta_{T_{(j)}}}^{\tau(F_{(j)})-\tau(F_{(j+1)})}\Bigg); \sqrt{2}-\prod_{j=1}^k\Big(\sqrt{2}-r_{\theta_{T_{(j)}}}\Big)^{\tau(F_{(j)})-\tau(F_{(j+1)})} \Bigg\rangle
\end{equation}

where the sequence $\{\theta_{T_{(j)}}\}_{j=0}^k$ indicates the indices permuted such that $\theta_{T_{(1)}}\preceq \theta_{T_{(2)}}\preceq\ldots \preceq \theta_{T_{(k)}}$, $F_{(j)}:=\{T_{(j)},T_{(j+1)},\ldots,T_{(k)}\}$ and $F_{(k+1)}=\emptyset$.
\item\textbf{Step IX:}  Calculate results of $CPM_i$  based on $\mbox{\textit{D-IFCGIO}}_q^\tau$ and $\mbox{\textit{D-IFCGIO}}_p^\tau$ for each alternative $P_i$ using aggregated decision matrix obtained in Step III and one of the operators
\begin{equation}
    CPM_i^q=\Bigg\langle \Bigg( \prod_{j=1}^k \mu_{\theta_{T_{(j)}}}^{\tau(F_{(j)})-\tau(F_{(j+1)})}, 1-\prod_{j=1}^k(1-\nu_{\theta_{T_{(j)}}})^{\tau(F_{(j)})-\tau(F_{(j+1)})} \Bigg) ; \prod_{j=1}^k r_{\theta_{T_{(j)}}}^{\tau(F_{(j)})-\tau(F_{(j+1)})} \Bigg\rangle
\end{equation}
or
\begin{equation}
   CPM_i^p= \Bigg\langle \Bigg( \prod_{j=1}^k \mu_{\theta_{T_{(j)}}}^{\tau(F_{(j)})-\tau(F_{(j+1)})}, 1-\prod_{j=1}^k(1-\nu_{\theta_{T_{(j)}}})^{\tau(F_{(j)})-\tau(F_{(j+1)})} \Bigg);\sqrt{2}-\prod_{j=1}^k\Big(\sqrt{2}-r_{\theta_{T_{(j)}}}\Big)^{\tau(F_{(j)})-\tau(F_{(j+1)})} \Bigg\rangle 
\end{equation}
where the sequence $\{\theta_{T_{(j)}}\}_{j=0}^k$ indicates the indices permuted such that $\theta_{T_{(1)}}\preceq \theta_{T_{(2)}}\preceq\ldots \preceq \theta_{T_{(k)}}$, $F_{(j)}:=\{T_{(j)},T_{(j+1)},\ldots,T_{(k)}\}$ and $F_{(k+1)}=\emptyset$.

\item\textbf{Step X:} Select the threshold number $\varepsilon\in [0,1]$ and find significance degree of each $P_i$ by utilizing one of formulas
\begin{equation}
   SD^q_i=\varepsilon CSM_i^q \oplus_{q} (1-\varepsilon) CPM_i^q 
\end{equation}

or 
\begin{equation}
 SD^p_i=\varepsilon CSM_i^p\oplus_{p} (1-\varepsilon) CPM_i^p   
\end{equation}

for $i=1,\ldots m$. 
\item\textbf{Step XI:} Compute the score values of the significance degrees $SD_i(i=1,\ldots m)$.
\item\textbf{Step XII:} The alternatives are ranked with the one with the highest score being the best alternative. If two alternatives have equal score values, the values of their accuracy functions may be taken into consideration.
\end{itemize}
\subsection{ A numerical example} \label{numerical} A university is planning to transition its energy sources to renewable ones by installing solar panels across its campus. This initiative aims to reduce carbon footprint, cut energy costs, and serve as an educational resource for students studying renewable energy technologies. The decision involves multiple criteria and alternatives, making it a complex MCGDM problem.

\begin{itemize}
    \item \textbf{Step I:} In the decision-making process, the university administration has appointed three engineers: two $E_1,E_2$ possess PhD degrees, while the third $E_3$ has a master's degree.  The evaluation process of solar panels includes five solar panels: monocrystalline silicon panel $(P_1)$, polycrystalline solar panel $(P_2)$, thin-film solar panel $(P_3)$, PERC solar panel $(P_4)$, bifacial solar panel $(P_5)$. There are four criteria: efficiency ($T_1$), cost ($T_2$), durability ($T_3$), installation complexity ($T_4$) for evaluating these panel types.
\item\textbf{Step II:} Three engineers assess the alternatives based on the linguistic terms presented in  Table \ref{tab:my_label11}. The evaluation results are shown in Table \ref{tab1333}. Then, by employing D-IFVs associated with each linguistic term, disc intuitionistic fuzzy group decision matrix can be constructed. Given that 
$T_2$ and $T_4$ represent cost criteria, we take the complement of these values. Consequently, this process results in the creation of the disc intuitionistic fuzzy group normalized decision matrix, as depicted in Table \ref{tab1444}.
    
    \begin{table}[h!]
    \centering
    \begin{tabular}{c|c|cccc}
    \hline
 
    Experts & Solar Panels & $T_1$ &$T_2$&$T_3$&$T_4$ \\ \hline
      & $P_1$ &EH&MH&EH & M\\
     & $P_2$ &VH &M&VH&M\\
  $E_1$  & $P_3$ &MH&L&ML&ML\\
     & $P_4$ &VH&MH&H&M\\
     & $P_5$ &H&MH&VH&MH\\ \hline
     & $P_1$ &VH&H&VH & MH\\
     & $P_2$ &MH&M&H&MH\\
  $E_2$  & $P_3$ &M&M&M&ML\\
     & $P_4$ &VH&MH&VH&M\\
     & $P_5$ &VH&H&VH&M
     \\\hline
      & $P_1$ &EH&VH&H & ML\\
     & $P_2$ &H&ML&MH&M\\
  $E_3$  & $P_3$ &ML&VL&ML&VL\\
     & $P_4$ &H&H&MH&ML\\
     & $P_5$ &H&VH&H&H
     \\\hline
    \end{tabular}
    \caption{ Group decision matrix with linguistic term}
    \label{tab1333}
\end{table}

\begin{table}[h!]
    \centering
    \begin{tabular}{c|c|cccc}
    \hline
 
    Engineers & Solar Panels & $T_1$ &$T_2$&$T_3$&$T_4$ \\ \hline
      & $P_1$ &$\langle (0.9,0.1);0.9\rangle$&$\langle (0.4,0.6);0.6\rangle$&$\langle (0.9,0.1);0.9\rangle$ & $\langle (0.5,0.5);0.5\rangle$\\
     & $P_2$ &$\langle (0.8,0.2);0.8\rangle$ &$\langle (0.5,0.5);0.5\rangle$&$\langle (0.8,0.2);0.8\rangle$&$\langle (0.5,0.5);0.5\rangle$\\
  $E_1$  & $P_3$ &$\langle (0.6,0.4);0.6\rangle$&$\langle (0.7,0.3);0.3\rangle$&$\langle (0.4,0.6);0.4\rangle$&$\langle (0.6,0.4);0.4\rangle$\\
     & $P_4$ &$\langle (0.8,0.2);0.8\rangle$&$\langle (0.4,0.6);0.6\rangle$&$\langle (0.7,0.3);0.7\rangle$&$\langle (0.5,0.5);0.5\rangle$\\
     & $P_5$ &$\langle (0.7,0.3);0.7\rangle$&$\langle (0.4,0.6);0.6\rangle$&$\langle (0.8,0.2);0.8\rangle$&$\langle (0.4,0.6);0.6\rangle$\\ \hline
     & $P_1$ &$\langle (0.8,0.2);0.8\rangle$&$\langle (0.3,0.7);0.7\rangle$&$\langle (0.8,0.2);0.8\rangle$ & $\langle (0.4,0.6);0.6\rangle$\\
     & $P_2$ &$\langle (0.6,0.4);0.6\rangle$&$\langle (0.5,0.5);0.5\rangle$&$\langle (0.7,0.3);0.7\rangle$&$\langle (0.4,0.6);0.6\rangle$\\
  $E_2$  & $P_3$ &$\langle (0.5,0.5);0.5\rangle$&$\langle (0.5,0.5);0.5\rangle$&$\langle (0.5,0.5);0.5\rangle$&$\langle (0.6,0.4);0.4\rangle$\\
     & $P_4$ &$\langle (0.8,0.2);0.8\rangle$&$\langle (0.4,0.6);0.6\rangle$&$\langle (0.8,0.2);0.8\rangle$&$\langle (0.5,0.5);0.5\rangle$\\
     & $P_5$ &$\langle (0.8,0.2);0.8\rangle$&$\langle (0.3,0.7);0.7\rangle$&$\langle (0.8,0.2);0.8\rangle$&$\langle (0.5,0.5);0.5\rangle$
     \\\hline
      & $P_1$ &$\langle (0.9,0.1);0.9\rangle$&$\langle (0.2,0.8);0.8\rangle$&$\langle (0.7,0.3);0.7\rangle$ & $\langle (0.6,0.4);0.4\rangle$\\
     & $P_2$ &$\langle (0.7,0.3);0.7\rangle$&$\langle (0.6,0.4);0.4\rangle$&$\langle (0.6,0.4);0.6\rangle$&$\langle (0.5,0.5);0.5\rangle$\\
  $E_3$  & $P_3$ &$\langle (0.4,0.6);0.4\rangle$&$\langle (0.8,0.2);0.2\rangle$&$\langle (0.4,0.6);0.4\rangle$&$\langle (0.8,0.2);0.2\rangle$\\
     & $P_4$ &$\langle (0.7,0.3);0.7\rangle$&$\langle (0.3,0.7);0.7\rangle$&$\langle (0.6,0.4);0.6\rangle$&$\langle (0.6,0.4);0.4\rangle$\\
     & $P_5$ &$\langle (0.7,0.3);0.7\rangle$&$\langle (0.2,0.8);0.8\rangle$&$\langle (0.7,0.3);0.7\rangle$&$\langle (0.3,0.7);0.7\rangle$
     \\\hline
    \end{tabular}
    \caption{Disc intuitionistic fuzzy group normalized decision matrix}
    \label{tab1444}
\end{table}
\item \textbf{Step III:} The weight vector of experts according to their experience is $\Omega=(\Omega_{1},\Omega_{2},\Omega_{3})=(0.4,0.4,0.2)$. After the disc intuitionistic fuzzy group normalized decision matrix is
aggregated by using operator $\mbox{\textit{D-IFWAO}}_q^A$, we obtain aggregated disc intuitionistic fuzzy decision matrix shown in Table \ref{tab1555}.

\begin{table}[h!]
    \centering
    \begin{tabular}{c|cccc}
    \hline
 
   Solar Panels  & $T_1$ &$T_2$&$T_3$&$T_4$ \\ \hline
       $P_1$ &$\langle (0.87,0.13);0.86\rangle$&$\langle (0.32,0.68);0.68\rangle$&$\langle (0.84,0.16);0.82\rangle$ &$\langle (0.49,0.51);0.51\rangle$\\
      $P_2$ &$\langle (0.71,0.29);0.69\rangle$&$\langle (0.52,0.48);0.48\rangle$&$\langle (0.73,0.27);0.72\rangle$&$\langle (0.46,0.54);0.54\rangle$\\
  $P_3$ &$\langle (0.53,0.47);0.51\rangle$&$\langle (0.66,0.34);0.34\rangle$&$\langle (0.44,0.56);0.44\rangle$&$\langle (0.65,0.35);0.35\rangle$\\
      $P_4$ &$\langle (0.78,0.22);0.78\rangle$&$\langle (0.38,0.62);0.62\rangle$&$\langle (0.73,0.27);0.72\rangle$&$\langle (0.52,0.48);0.48\rangle$\\
      $P_5$ &$\langle (0.75,0.25);0.74\rangle$&$\langle (0.32,0.68);0.68\rangle$&$\langle (0.78,0.22);0.78\rangle$&$\langle (0.42,0.58);0.58\rangle$\\ \hline

    \end{tabular}
    \caption{Aggregated disc intuitionistic fuzzy decision matrix}
    \label{tab1555}
\end{table}
\item \textbf{Step IV:} The  importance weights of the criteria are assessed by engineers with assistance from Table \ref{tab:my_label11}, leading to the importance weights matrix of criteria listed in Table \ref{tab:my_label6666}.

\item \textbf{Step V:}  By utilizing weight vector of experts $\Omega=(\Omega_{1},\Omega_{2},\Omega_{3})=(0.4,0.4,0.2)$  and operator $DIFWA_p^A$, the importance weights matrix of criteria given in Table \ref{tab:my_label6666} is aggregated. Aggregated importance weights matrix of criteria is also shown in Table \ref{tab:my_label6666}.

\begin{table}[]
    \centering
    \begin{tabular}{|c|ccc|c|}
    \hline
     Criteria & $E_1$ & $E_2$ & $E_3$& Aggregated results \\ \hline
       $T_1$  & $\langle (0.8,0.2);0.8\rangle$ & $\langle (0.7,0.3);0.7\rangle$ &$\langle (0.5,0.5);0.5\rangle$& $\langle (0.72,0.28);0.69 \rangle$ \\
       $T_2$  & $\langle (0.6,0.4);0.6\rangle$ & $\langle (0.6,0.4);0.6\rangle$ &$\langle (0.4,0.6);0.4\rangle$&$\langle (0.57,0.43);0.55 \rangle$ \\
       $T_3$  & $\langle (0.6,0.4);0.6\rangle$ & $\langle (0.5,0.5);0.5\rangle$ &$\langle (0.3,0.7);0.3\rangle$&$\langle (0.51,0.49);0.49 \rangle$ \\
       $T_4$  & $\langle (0.5,0.5);0.5\rangle$ & $\langle (0.4,0.6);0.4\rangle$ &$\langle (0.2,0.8);0.2\rangle$&$\langle (0.41,0.59);0.38 \rangle$ \\ \hline
    
    \end{tabular}
    \caption{The importance weights matrix of criteria and its aggregated results}
    \label{tab:my_label6666}
\end{table}

\item \textbf{Step VI:} After computing the score value of each criterion $T_j$ via $\mathcal{S}$ introduced in Definition \ref{scoref} for $\xi=0.8$, we normalize score of each criterion $T_j$ using Eq. \ref{eq1}. Therefore the weight vector of the criteria given in is $\omega=(\omega_1,\omega_2,\omega_3,\omega_4)=(0.326,0.258,0.232,0.184)$.

\item\textbf{Step VII:} A fuzzy measure employing the $\lambda$-fuzzy methodology can be formulated using Takahagi's algorithm.  To specify the construction, let's assign $\lambda$ a value of $0.5$ and consider weights $\omega=(\omega_1,\omega_2,\omega_3,\omega_4)=(0.326,0.258,0.232,0.184)$. This yields the fuzzy measure $\tau$, depicted in Table \ref{tab:my_label777}.

\begin{table}[h!]
    \centering
    \begin{tabular}{lll}
    \hline
       $\tau(\varnothing)=0$  & $\tau(\{T_1\})=0.282631$ &$\tau(\{T_2\})=0.220555$    \\
       $\tau(\{T_3\})=0.197269$  & $\tau(\{T_4\})=0.154918$ &$\tau(\{T_1,T_2\})=	0.534354$    \\
       $\tau(\{T_1,T_3\})=0.507777$  & $\tau(\{T_1,T_4\})=0.459442$ &$\tau(\{T_2,T_3\})=0.439578$    \\
       $\tau(\{T_2,T_4\})=	0.392557$  & $\tau(\{T_3,T_4\})=0.367467$ &$\tau(\{T_1,T_2,T_3\})=0.784329$    \\
       $\tau(\{T_1,T_2,T_4\})=0.730663$&$\tau(\{T_1,T_3,T_4\})=0.702027$ & $\tau(\{T_2,T_3,T_4\})=	0.628546$\\
       $\tau(\{T_1,T_2,T_3,T_4\})=1$\\ \hline
        
    \end{tabular}
    \caption{A fuzzy measure}
    \label{tab:my_label777}
\end{table}
\item \textbf{Step VIII-IX:} By computing $CSM_i$ and $CPM_i$ for each $P_i$, the results of $CSM_i$ and $CPM_i$ are obtained and are summarized in Table \ref{tab18888}.

\begin{table}[h!]
    \centering
    \begin{tabular}{|c|cccc|}
    \hline
 
   Solar Panels  & $CSM^q$ &$CSM^p$&$CPM^q$&$CPM^p$ \\ \hline
       $P_1$ &$\langle (0.71,0.29);0.72\rangle$&$\langle (0.71,0.29);0.74\rangle$&$\langle (0.57,0.43);0.72\rangle$ &$\langle (0.57,0.43);0.74\rangle$\\
      $P_2$ &$\langle (0.62,0.38);0.59\rangle$&$\langle (0.62,0.38);0.62\rangle$&$\langle (0.59,0.40);0.59\rangle$&$\langle (0.59,0.40);0.62\rangle$\\
  $P_3$ &$\langle (0.56,0.44);0.42\rangle$&$\langle (0.56,0.44);0.43\rangle$&$\langle (0.55,0.45);0.42\rangle$&$\langle (0.55,0.45);0.43\rangle$\\
      $P_4$ &$\langle (0.63,0.37);0.65\rangle$&$\langle (0.63,0.37);0.67\rangle$&$\langle (0.57,0.43);0.65\rangle$&$\langle (0.57,0.43);0.67\rangle$\\
      $P_5$ &$\langle (0.61,0.39);0.69\rangle$&$\langle (0.61,0.39);0.70\rangle$&$\langle (0.52,0.48);0.69\rangle$&$\langle (0.52,0.48);0.70\rangle$\\ \hline

    \end{tabular}
    \caption{The results of $CSM_i$ and $CPM_i$}
    \label{tab18888}
\end{table}
\item\textbf{Step X:} If we select $\varepsilon=0.3$, significance degrees $SD_i$ of each $P_i$ are listed in Table \ref{tab19999}. 
\begin{table}[h!]
    \centering
    \begin{tabular}{|c|cc|}
    \hline
 
   Solar Panels  & $SD^q$ &$SD^p$ \\ \hline
       $P_1$ &$\langle (0.619,0.380);0.717\rangle$&$\langle (0.619,0.380);0.741\rangle$\\
      $P_2$ &$\langle (0.601,0.399);0.596\rangle$&$\langle (0.601,0.399);0.620\rangle$\\
  $P_3$ &$\langle (0.553,0.447);0.420\rangle$&$\langle (0.553,0.447);0.429\rangle$\\
      $P_4$ &$\langle (0.593,0.407);0.651\rangle$&$\langle (0.593,0.407);0.668\rangle$\\
      $P_5$ &$\langle (0.553,0.447);0.695\rangle$&$\langle (0.553,0.447);0.702\rangle$\\ \hline

    \end{tabular}
    \caption{The results of significance degrees $SD_i$}
    \label{tab19999}
\end{table}
\item\textbf{Step XI:} Score values of the significance degrees $SD_i(i=1,\ldots m)$ are attained by utilizing $\mathcal{S}$ introduced in Definition \ref{scoref} for $\xi=0.8$. The results are shown Table \ref{tab10000}. 

\begin{table}[h!]
    \centering
    \begin{tabular}{|c|cc|}
    \hline
 
   Solar Panels  & $\mathcal{S}(SD^q)$ &$\mathcal{S}(SD^p)$ \\ \hline
       $P_1$ &$0.596$&$0.600$\\
      $P_2$ &$0.565$&$0.569$\\
  $P_3$ &$0.501$&$0.502$\\
      $P_4$ &$0.567$&$0.568$\\
      $P_5$ &$0.540$&$0.541$\\ \hline

    \end{tabular}
    \caption{The score values of significance degrees $SD^q$ and $SD^p$ }
    \label{tab10000}
\end{table}
\item\textbf{Step XII:} For score values of significance degree $SD^q$, the order is established as $P_1>P_4>P_2>P_5>P_3$, whereas for score values associated with significance degree $SD^p$, the sequence is determined as $P_1>P_2>P_4>P_5>P_3$. Both $SD^q$ and $SD^p$ determine $P_1$  as the most favorable solar panel. Monocrystalline silicon panel emerges as the best alternative as a result of the decision-making process. The difference in the ranking of $P_2$ and $P_4$ solar panels is a result of the interactions between the criteria and the radius change due to the circular structure
\end{itemize}

\section{Assessment of performance} \label{performance}

In this section, we conduct sensitivity analysis and validity assessment to observe how variations in different scenarios of the solar panel selection problem may impact the ranking of alternatives. 

\subsection{Sensitivity analysis}
Sensitivity analysis is a technique used to assess how the variation (or uncertainty) in the output of a mathematical model or system can be attributed to variations in the inputs. The main goal of sensitivity analysis is to understand which input parameters have the most significant impact on the output of the model. This analysis helps in identifying critical factors and understanding the robustness of the model.
In this subsection, we  concentrate on an analysis of the effect of the variation of the parameters ($\varepsilon$ and $\xi$) considered in the proposed D-IF CASPAS method on the ranking of alternatives. The study of these parameters reveals the impact of different aspects of D-IFVs and the CASPAS method on the ranking results of the solar panel selection problem.  
\begin{itemize}
 \item \textbf{Analysis of parameter $\varepsilon$:} In Step X of the proposed CASPAS method given in Sub-section \ref{waspasch}, the $CSM^q$, $CSM^p$, $CPM^q$, and $CPM^q$ values are calculated as D-IFVs by adjusting the threshold number $\varepsilon\in [0,1]$. $\varepsilon$ affects the results of the significance degrees $SD^q$ and $SD^p$ obtained from $CSM^q$, $CSM^p$, $CPM^q$, and $CPM^q$. As $\varepsilon$ grows larger, the impact of $CSM^q$ on $SD^q$ becomes more pronounced, whereas the influence of $CPM^q$ diminishes. A similar scenario applies to $SD^p$ as well. In Step X of the proposed CASPAS method, an analysis is applied by varying $\varepsilon$ between $0.1$ and $0.9$ to examine the effect of $\varepsilon$ on the results. Figure \ref{epsilon} shows the results of the ranking of the $\varepsilon$ parameter. These results are calculated according to the score values $\xi=0.8$. The ranking results can be respectively expressed as follows. For the significance degree $SD^q$, two different rankings are obtained. For $\varepsilon=0.1$ and $0.2$  we find the order $P_1>P_2>P_4>P_5>P_3$, while by varying $\varepsilon$ from $0.3$ to $0.7$ we get the order $P_1>P_4>P_2>P_5>P_3$. If $\varepsilon$ equals $0.8$ and $0.9$, $P_1>P_4>P_5>P_2>P_3$ is attained. For the significance degree $SD^p$, we get also two different rankings. When $\varepsilon$ changes from $0.1$ and $0.3$, the ordering of alternatives remains stable as $P_1>P_2>P_4>P_5>P_3$. Conversely, the ranking shifts to the ranking $P_1>P_4>P_2>P_5>P_3$ as $\varepsilon$ increases from $0.4$ to $0.9$. Upon examining the ranking of alternatives, it becomes evident that $P_3$ consistently emerges as the least favorable alternative and $P_1$ also appears to be the most favorable alternative across all parameters $\varepsilon$ ranging from $0.1$ to $0.9$ for both $SD^q$ and $SD^p$. Nonetheless, when the results are analysed, the order of the second, third and fourth ranked alternatives varies according to the epsilon parameter. This is mainly due to the adjustment of the effect of the $CSM$ and $CPM$ in the proposed CASPAS method on the significance $SD$ with the $\varepsilon$ parameter. The significant role of the radius in D-IFSs  and D-IFVs in altering these rankings is noteworthy. With adjustments in the radius degree, there can be shifts in the alternatives because the score values are affected. Consequently, the use of D-IFSs and D-IFVs in decision-making leads to more accurate modeling compared to assigning fixed values to criteria, thanks to the adaptable nature of their circular framework. 

 \begin{figure}[h!]
    \centering
    \includegraphics[scale=0.55]{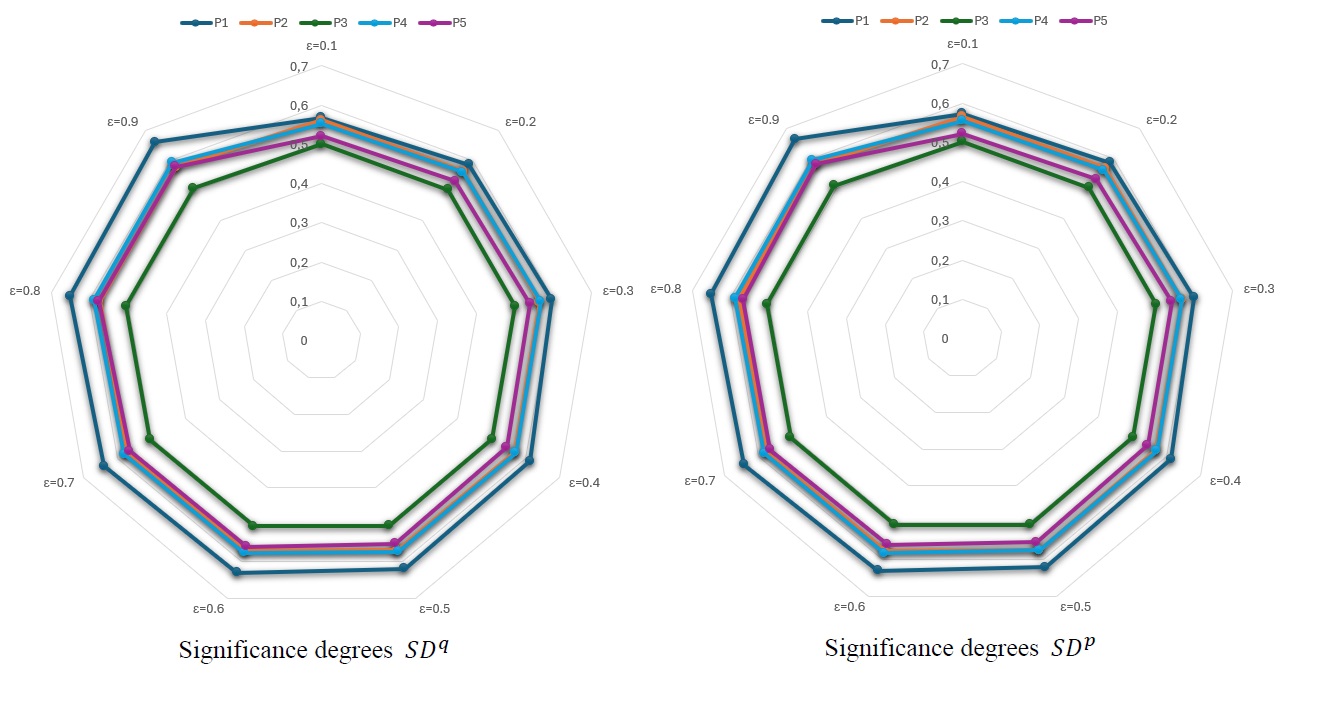}
    \caption{Sensitivity analysis across $\varepsilon$ values ranging from $0.1$ to $0.9$}
    \label{epsilon}
\end{figure}

    \item \textbf{Analysis of parameter $\xi$:} In Step XI of the CASPAS method proposed in Sub-section \ref{waspasch}, the score function proposed in Definition \ref{scoref} is used to calculate the score values of the significance degrees $SD^q$ and $SD^p$. The most important advantage of this score function is that the effect of the radius degree varies depending on the decision maker's choice of the parameter $\xi$ between 0 and 1. Thus, the effect of the radius in D-IFVs can be adjusted and its effect on the ranking can be analyzed. As $\xi$ approaches $1$, the degree of radius becomes essentially irrelevant in the decision-making process. To assess the impact of the parameter
$\xi$ as detailed in Step XI of the proposed CASPAS methodology, we perform an analysis by varying
$\xi$ from $0.1$ to $0.9$ for the score values of the significance degrees $SD^q$ and $SD^p$. The radar chart for the variation of the parameter $\xi$ between $0.1$ and $0.9$ is shown in Figure \ref{radarchart}. The findings obtained from Figure \ref{radarchart} can be summarised as follows. According to significance degree $SD^q$, there are four different ranking. The
consistent ranking of alternatives is $P_1>P_5>P_4>P_2>P_3$ from $\xi=0.1$ and $0.4$. When $\xi$ equals $0.5$, alternatives are ordered as $P_1>P_4>P_5>P_2>P_3$.  When $\xi$ varies from $0.6$ and $0.8$, the ranking of alternatives is obtained as $P_1>P_4>P_2>P_5>P_3$. Lastly the ranking $P_1>P_2>P_4>P_5>P_3$ is obtained for $\xi=0.9$.  According to significance degree $SD^p$, we reach four different rankings. If $\xi$ varies between $0.1$ and $0.3$, we attain the ranking of alternatives as $P_1>P_5>P_4>P_2>P_3$. For $\xi=0.4$ and $0.5$ the ranking is $P_1>P_4>P_5>P_2>P_3$. When $\xi$ is set to $0.6$ and $0.7$, the
consistent ranking of alternatives is obtained as $P_1>P_4>P_2>P_5>P_3$. Lastly we get ranking $P_1>P_2>P_4>P_5>P_3$ for $\xi=0.8$ and $0.9$. When the ranking of the alternatives is analysed, it is seen that $P_3$ is the worst alternative and $P_1$ is the best alternative for each parameter $\xi$ from $0.1$ and $0.9$ according to $SD^q$ and $SD^p$. However, the ranking of the second, third and fourth alternatives changes. For instance, according to $SD^q$ when $\xi=0.1$, $P_5$ is the second  best alternative, while for $\xi=0.5$, $P_4$ becomes the second best alternative. This variation results from the circular structures of D-IFSs and D-IFVs. Since the  parameter $\xi$ in the score function is adjustable, the impact of the radius degree on rankings also varies. This causes the order of alternatives to change. Additionally, one of the conclusions that can be drawn from this analysis is the potential for different rankings when the degrees of membership and non-membership are modeled with a radius degree instead of being precisely assigned. This highlights the significance of decision-making processes involving D-IFSs constructed with circular structures.  

\begin{figure}[h!]
    \centering
    \includegraphics[scale=0.55]{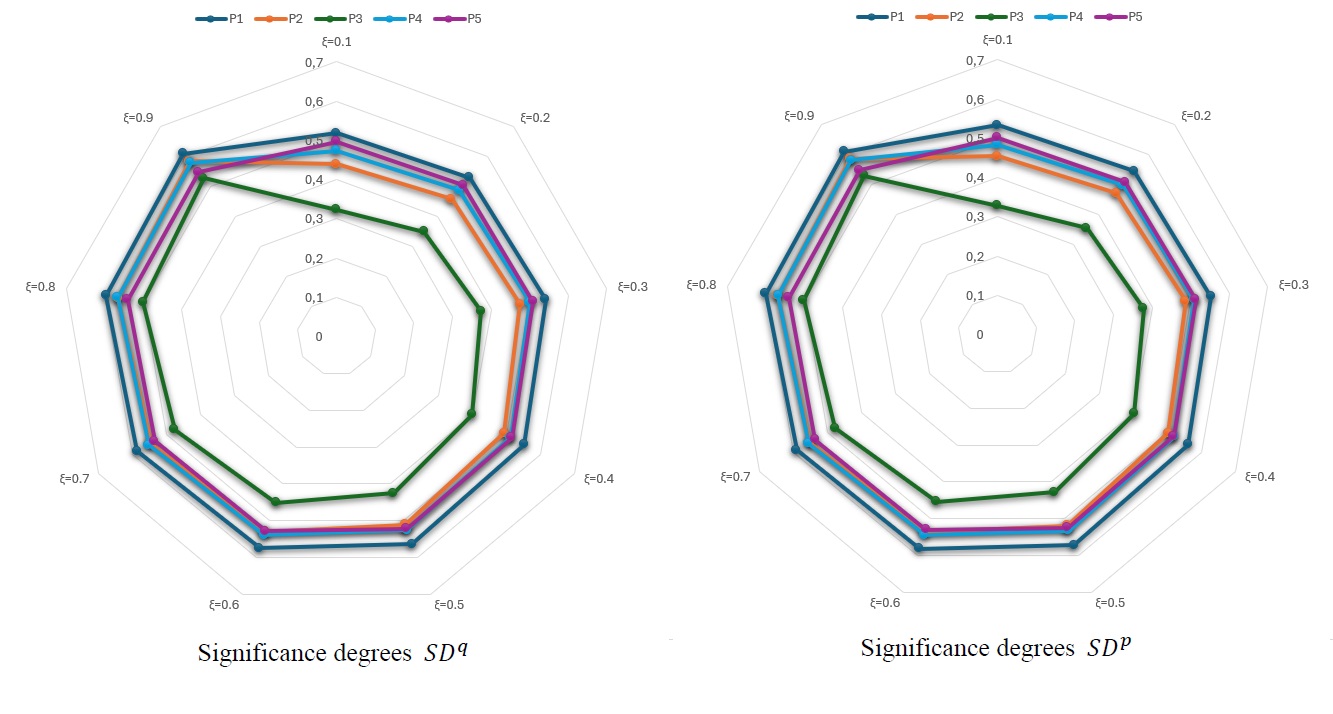}
    \caption{Sensitivity analysis across a range of $\xi$ values spanning from $0.1$ to $0.9$.}
    \label{radarchart}
\end{figure}
\end{itemize}

\subsection{Validity analysis}
Validity analysis in  MCDM and MCGDM involves assessing the reliability, relevance, and effectiveness of the decision-making process and the associated outcomes.
The efficacy of MCDM method is shaped by various factors such as the interplay between alternatives, the consistency of criteria, and impartial evaluations among decision-makers. Wang and Triantaphyllou \cite{WangX} introduced three evaluative criteria to validate the presented method in MCDM. These three test conditions are systematically applied to the proposed method as follow:

\begin{itemize}
    \item \textbf{I.Condition:} When replacing a non-optimal alternative with a worse one without altering the significance of each criterion, the best alternative should remain unchanged in a useful MCDM method.

    \item \textbf{II.Condition:} An effective MCDM method should exhibit the transitive property. In other words, if alternative $X$ is favored over alternative $Y$, and alternative $Y$ is chosen over alternative $Z$, then it follows that alternative $X$ should be selected over alternative $Z$.

     \item \textbf{III.Condition:} If a decision-making problem is broken down into multiple sub-problems, the order of the subproblems  through the decision-making approach should align with the ranking outcome of the original problem.
\end{itemize}
To assess the validity of the proposed CASPAS method, the following examinations are carried out.

\begin{itemize}
   \item \textbf{I.Condition:} In order to fulfil this condition, the non-optimal alternative $P_3$ is substituted with an arbitrary worse option $P'_3$, and the evaluation of $P'_3$ is presented in Table \ref{tab:my_label121211} based on the linguistic terms presented  in Table \ref{tab:my_label11}.

   \begin{table}[h!]
       \centering
       \begin{tabular}{|c|cccc|}
       \hline
         Engineers &  $T_1$ & $T_2$ & $T_3$ &$T_4$  \\ \hline
          $E_1$  &  M & L & L &ML  \\ \hline
         $E_2$  & ML & M & ML &M  \\ \hline
         $E_3$  & L & EL & L &VL  \\ \hline
          
       \end{tabular}
       \caption{The evaluation of $P'_3$ based on the linguistic terms}
       \label{tab:my_label121211}
   \end{table}
We can now explore if there are any shifts in the top-ranked alternative for $\varepsilon=0.3$ and $\xi=0.5$. Utilizing the proposed CASPAS method, the hierarchy of alternatives is determined as $P_1>P_5>P_4>P_2>P'_3$, with $P_1$ retaining its position as the best alternative. Therefore, I.Conditon satisfies for the proposed method.
\item \textbf{II. and III. Condition:} The same solar panel problem is divided into four sub-problems, which include $\{P_1, P_2, P_3,P_4\}, \{P_1, P_2, P_4,P_5\}, \{P_2, P_3, P_4,P_5\}$, and $\{P_1, P_2, P_3, P_5\}$, respectively. By using the proposed WASPAS method to solve these sub-problems, the ranking orders
of these sub-problems can be obtained as: $P_1>P_4>P_2>P_3$, $P_1>P_5>P_4>P_2$, $P_5>P_4>P_2>P_3$, and $P_1>P_5>P_2>P_3$. Upon merging the results of  these sub-problems, the comprehensive result is derived as: $P_1>P_5>P_4>P_2>P_3$ which is consistent with with the result of the same problem. Therefore, the proposed CASPAS
method is confirmed for I.Condition and  II.Condition.
\end{itemize}

\section{Comparative Analysis} \label{comparative}

The section commences with a comparison between weighted aggregation and Choquet integral operators, then proceeds to conduct a thorough analysis of the proposed CASPAS method in relation to the existing C-IF-TOPSIS and C-IF-VIKOR methods. Ultimately, it assesses the advantages of the newly proposed method over these existing approaches.

\subsection{A comparison of weighted aggregation operators with Choquet integral operators}

In this subsection, we discuss the comparison between two widely used aggregation operators: weighted sum and Choquet integrals. Our aim is to determine the optimal alternative among a variety of alternatives, each assessed based on distinct criteria. Despite the widespread popularity and simplicity of the weighted sum approach, the practical application of the Choquet integral remains challenging. However, theoretically, the Choquet integral promises superior outcomes aligning more closely with a decision maker's preferences. Meyer and Pirlot \cite{Meyer} have demonstrated that the Choquet integral can consider significantly more preferences compared to the weighted arithmetic mean, with a notable difference especially evident when the number of criteria is high. Furthermore, Lust \cite{Lust} has shown that as the number of criteria increases, the use of the Choquet integral instead of the weighted arithmetic mean increases the likelihood of achieving a more optimal ranking. When using a weighted arithmetic mean, the interaction between decision maker priorities and criteria cannot be assessed with sufficient precision. In order to make a comparative analysis of the results of weighted operators and Choquet integral operators, in the proposed CASPAS method, in Steps VIII and IX, the ignificance  degrees $SD$ are calculated using $WSM$ and $WPM$ defined by weighted operators instead of $CSM$ and $CPM$ defined by Choquet integral operator.  $WSM$ and $WPM$ are defined as follows. For each alternative $P_i (i=1,\ldots m$), we have

\begin{equation}
    WSM^q_i=\Bigg\langle \Bigg( 1-\prod_{j=1}^k(1-\mu_{\theta_{T_{j}}})^{\omega_j}, \prod_{j=1}^k\nu_{\theta_{T_{j}}}^{\omega_j}\Bigg); \prod_{j=1}^k r_{\theta_{T_{j}}}^{\omega_j} \Bigg\rangle
\end{equation}
or 

\begin{equation}
    WSM^p_i=\Bigg\langle \Bigg(1-\prod_{j=1}^k(1-\mu_{\theta_{T_{j}}})^{\omega_j}, \prod_{j=1}^k \nu_{\theta_{T_{j}}}^{\omega_j}\Bigg); \sqrt{2}-\prod_{k=1}^m\Big(\sqrt{2}-r_{\theta_{T_{j}}}\Big)^{\omega_j} \Bigg\rangle
\end{equation}
and

\begin{equation}
    WPM^q_i=\Bigg\langle \Bigg( \prod_{j=1}^k \mu_{\theta_{T_{j}}}^{\omega_j}, 1-\prod_{j=1}^k(1-\nu_{\theta_{T_{j}}})^{\omega_j} \Bigg) ; \prod_{j=1}^k r_{\theta_{T_{j}}}^{\omega_j} \Bigg\rangle
\end{equation}
or

\begin{equation}
    WPM^p_i=\Bigg\langle \Bigg( \prod_{j=1}^k \mu_{\theta_{T_{j}}}^{\omega_j}, 1-\prod_{j=1}^k(1-\nu_{\theta_{T_{j}}})^{\omega_j} \Bigg) ; \sqrt{2}-\prod_{k=1}^m\Big(\sqrt{2}-r_{\theta_{T_{j}}}\Big)^{\omega_j} \Bigg\rangle
\end{equation}

The results of the significance degrees $SD$ based on $WSM$ and $WPM$  calculated for $\varepsilon=0.1$ are listed in Table \ref{tab:my_label3343} according to $q$ and $p$. Table \ref{tab:my_label3343} also shows the score values which are obtained by utilizing $\mathcal{S}$ introduced in Definition \ref{scoref} for $\xi=0.9$.

\begin{table}[h!]
    \centering
    \begin{tabular}{|c|cc|cc|}
    \hline
      Solar Panels   & $SD^q$ & $\mathcal{S}(SD^q)$ & $SD^p$ & $\mathcal{S}(SD^p)$ \\ \hline
         $P_1$   & $\langle (0.60,0.39);0.73\rangle$ & $0.596$ & $\langle (0.60,0.39);0.75\rangle$ & $0.598$ \\
         $P_2$   & $\langle (0.62,0.38);0.61\rangle$ & $0.601$ & $\langle (0.62,0.38);0.62\rangle$ & $0.602$ \\
         $P_3$   & $\langle (0.56,0.43);0.42\rangle$ & $0.538$ & $\langle (0.56,0.43);0.42\rangle$ & $0.538$ \\
         $P_4$   & $\langle (0.60,0.40);0.66\rangle$ & $0.587$ & $\langle (0.60,0.40);0.68\rangle$ & $0.588$ \\
         $P_5$   & $\langle (0.56,0.44);0.70\rangle$ & $0.587$ & $\langle (0.56,0.44);0.71\rangle$ & $0.554$ \\ \hline
    \end{tabular}
    \caption{The results of significance degrees $SD$ based on $WSM$ and $WPM$}
    \label{tab:my_label3343}
\end{table}

When the results in Table \ref{tab:my_label3343} are analysed, it is seen that the ranking $P_2>P_1>P_4>P_5>P_3$ is obtained according to $SD$ based on $WSM$ and $WPM$. On the other hand, according to the Choquet integral operator based CASPAS method proposed in Subsection \ref{waspasch}, we get the ranking $P_1>P_2>P_4>P_3>P_5$ by taking $\varepsilon=0.1$ and $\xi=0.9$. As can be seen from the results, according to the weighted aggregation operators, the most favorable solar panel is $P_2$ and the most unfavorable solar panel is P3, while according to the Choquet integral operator, the most favorable solar panel is $P_1$ and the most unfavorable solar panel is $P_5$. The main reason for this difference is that Choquet integral operators take into account the interaction between criteria. Since weighted aggregation operators do not consider the interaction between criteria, more precise results can be obtained with Choquet integral operators  in decision-making applications.  

\subsection{A comparison with C-IF TOPSIS}

In this section, we employ C-IF TOSPIS to use in the solar panel selection, conducting a thorough comparative analysis with our proposed approach.
TOPSIS (Technique for Order of Preference by Similarity to Ideal Solution) is a MCDM used to determine the best alternative from a set of options. It aims to identify the alternative that is closest to the ideal solution based on a set of criteria. TOPSIS, pioneered by Hwang and Yoon \cite{Hwang} as a method in MCDM, was crafted with the explicit goal of addressing shortcomings in prevailing decision-making methodologies. Its purpose is to offer a more intuitive and pragmatic framework, mitigating the subjectivity inherent in criterion weighting and enhancing result interpretation. By furnishing decision-makers with a structured and quantitative methodology, TOPSIS empowers them to systematically assess and rank alternatives across various criteria. Chen \cite{Chen} proposed evolved C‑IF Minkowski distance measures for C-IFSs and C-IFVs, which are defined by 
\begin{equation}
    \mathcal{D}_{(3)}^\beta(\theta,\kappa)=\frac{1}{2}\Bigg\{\frac{1}{\sqrt{2}}|r_\theta-r_\kappa|+\Bigg[\frac{1}{2}\Bigg(|\mu_\theta-\mu_\kappa|^\beta+|\nu_\theta-\nu_\kappa|^\beta\Bigg)\Bigg]^\frac{1}{\beta}\Bigg\},
\end{equation}
for two C-IFVs $\theta=\langle (\mu_\theta,\nu_\theta);r_\theta\rangle$ and $\kappa=\langle (\mu_\kappa,\nu_\kappa);r_\kappa\rangle$. Here, $\beta$ is a positive integer and represents the metric parameter. Later, Chen \cite{Chen} introduced a C-IF TOPSIS methodology based on C-IF Minkowski distance measures and applied it to a site selection issue of large-scale epidemic hospitals.  To put it briefly, in this C-IF TOPSIS method, after the aggregation process, the displaced ideal and anti-ideal ratings are derived with 
\begin{equation}
\theta_{P_*}=\Big\langle \max_{1\leq i \leq m} \mu_{\theta_{P_i}},\min_{1\leq i \leq m}\nu_{\theta_{P_i}}; \max_{1\leq i \leq m} r_{\theta_{P_i}} \Big \rangle,
\end{equation}
and
\begin{equation}
\theta_{P_\neg}=\Big\langle \min_{1\leq i \leq m} \mu_{\theta_{P_i}},\max_{1\leq i \leq m}\nu_{\theta_{P_i}}; \max_{1\leq i \leq m} r_{\theta_{P_i}} \Big \rangle,
\end{equation}
respectively. Then the relative closeness
coefficient is calculated with 
\begin{equation}
    \mathcal{R}_{(3)}^{\beta*}(P_i)=\frac{\mathcal{D}_{(3)}^\beta(\theta_{P_i},\theta_{P_\neg})}{\mathcal{D}_{(3)}^\beta(\theta_{P_i},\theta_{P_*})+\mathcal{D}_{(3)}^\beta(\theta_{P_i},\theta_{P_\neg})},
\end{equation}
for $i=1,\ldots,m$. The alternative with the highest relative closeness
coefficient is the best alternative.

Let us apply C-IF-TOPSIS to solve the solar panel selection problem discussed in Sub-section \ref{numerical}. If the disc intuitionistic decision matrix given in Table \ref{tab1555} is aggregated with  operator  $\mbox{\textit{D-IFWGO}}_q$ with the weight vector $w=(0.326,0.258,0.232,0.181)$, the results listed in Table \ref{tab:my_label34533} are obtained with the displaced ideal and anti-ideal ratings.

\begin{table}[h!]
    \centering
    \begin{tabular}{|c|c|}
    \hline
         & $\mbox{\textit{D-IFWGO}}_q$ \\ \hline
      $P_1$   & $\langle (0.59,0.41);0.73 \rangle$\\
       $P_2$   & $\langle (0.61,0.39);0.61 \rangle$\\
        $P_3$   & $\langle (0.56,0.44);0.42 \rangle$\\
         $P_4$   & $\langle (0.59,0.41);0.66 \rangle$\\
          $P_5$   & $\langle (0.55,0.45);0.70 \rangle$\\ \hline
           $P_*$   & $\langle (0.61,0.39);0.73 \rangle$\\
           $P_\neg$   & $\langle (0.55,0.45);0.73 \rangle$\\ \hline
           
    \end{tabular}
    \caption{Aggregation results with the displaced ideal and anti-ideal ratings}
    \label{tab:my_label34533}
\end{table}

After the calculation for $\beta=3$, we find the relative closeness
coefficients $\mathcal{R}_{(3)}^{\beta*}(P_1)=0.667, \mathcal{R}_{(3)}^{\beta*}(P_2)=0.631, \mathcal{R}_{(3)}^{\beta*}(P_3)=0.459, \mathcal{R}_{(3)}^{\beta*}(P_4)=0.562,$ and $\mathcal{R}_{(3)}^{\beta*}(P_5)=0.207$. Therefore, the ranking $P_1>P_2>P_4>P_3>P_5$ is attained. This ranking is consistent with the proposed CASPAS method. In other words, there is consistency between the results of C-IF-TOPSIS and D-IF-CASPAS methods. 
\subsection{A comparison with C-IF VIKOR}
The VIKOR (VIseKriterijumska Optimizacija I Kompromisno Resenje) method is utilized in MCDM and MCGDM, pioneered by Opricovic \cite{Opricovic}. Its main goal is to pinpoint a compromise solution that serves as a practical alternative closely resembling the ideal solution, followed by establishing a ranking of these alternatives, considering various evaluation criteria. Subsequently, a variation of the VIKOR method, known as the fuzzy VIKOR method, was devised to handle MCDM challenges in fuzzy environments \cite{Opricovic1}. This approach emphasizes the proximity of alternatives to the ideal solution, effectively managing decision-making hurdles in fuzzy or uncertain contexts. The VIKOR approach determines a compromise solution based on mutual concessions. Kahraman and Otay \cite{Kahraman} proposed an extended VIKOR method for C-IFSs. Then, the waste disposal location selection problem was investigated using this C-IF VIKOR. This method can be briefly summarised as follows. Firstly, the decision matrices of the experts are aggregated and then the aggregated decision matrix is created. C-IF positive and negative ideal solutions are determined as
\begin{equation}
    \theta^*_j= \max_{1\leq i \leq m} \theta_{i,j} \mbox{ and }  \theta^-_j= \max_{1\leq i \leq m} \theta_{i,j}
\end{equation}
using the relative score function (RSF). Maximum group utility, minimum individual regret and VIKOR index are calculated by using $D_p$ and $D_o$ distances, which are defined in two different ways as pessimistic and optimistic view.

Now, we will implement the C-IF VIKOR method to address the solar panel selection issue outlined in Sub-section \ref{numerical}. Firstly, using the aggregated decision matrix given in Table \ref{tab1555}, the positive and negative ideal solutions shown in Table \ref{1235432} are determined according to the RSF.

\begin{table}[h!]
    \centering
    \begin{tabular}{|c|cccc|}
    \hline
         &  $T_1$ & $T_2$ & $T_3$ & $T_4$  \\\hline
       $\theta^*$  & $\langle (0.87,0.13);0.86 \rangle$ & $\langle (0.66,0.34);0.34 \rangle$& $\langle (0.84,0.16);0.82 \rangle$ & $\langle (0.65,0.35);0.35 \rangle$ \\ \hline
       $\theta^-$  & $\langle (0.53,0.47);0.51 \rangle$ & $\langle (0.32,0.68);0.68 \rangle$& $\langle (0.44,0.56);0.44 \rangle$ & $\langle (0.42,0.58);0.58 \rangle$ \\ \hline
    \end{tabular}
    \caption{C-IF positive and negative ideal solutions}
    \label{1235432}
\end{table}

By following the necessary steps, according to the VIKOR index, we reach the rankings $P_1>P_2>P_4>P_5>P_3$ for the optimistic case and $P_1>P_2>P_5>P_4>P_3$ for the pessimistic case. When compared with the proposed D-IF CASPAS method, it is seen that the results are compatible with C-IF VIKOR. 

In this section where the D-IF CASPAS method is compared to other approaches, the similarity in results between the proposed method and C-IF TOPSIS and C-IF VIKOR demonstrates result consistency. Furthermore, contrasting the weighted summation operators with the Choquet integral operator highlights how criteria interactions can influence outcomes. Despite its more intricate calculation process compared to weighted aggregation operators, Choquet integral operators yield more precise results by considering criterion interactions, facilitating more accurate modeling of real-world problems. Figure \ref{comparisonmethodsd-ıf} summarises the comparison of the results of the existing methods and the method proposed in this study.

\begin{figure}[h!]
    \centering
    \includegraphics[scale=0.55]{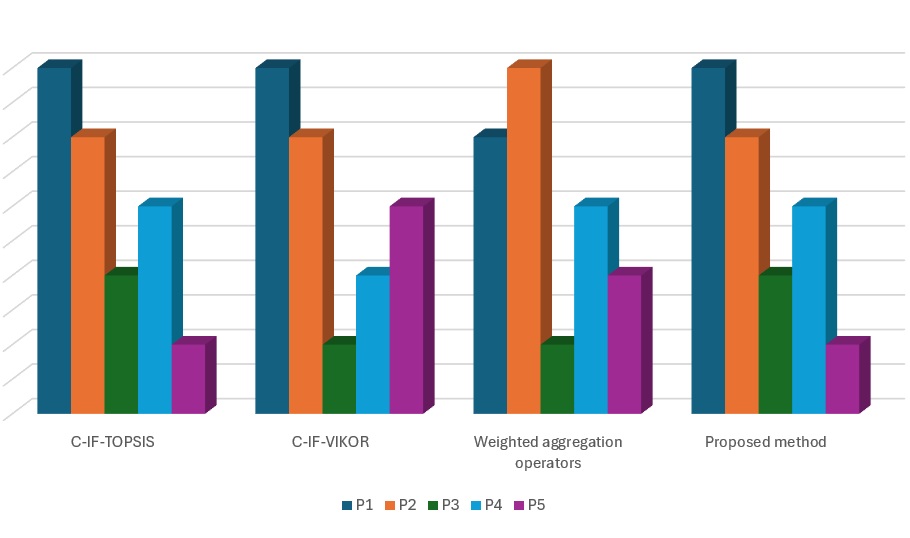}
    \caption{Ranking results of existing methods for the solar panel selection problem}
    \label{comparisonmethodsd-ıf}
\end{figure}

\section{Contributions and future studies}\label{confut}
This study proposed D-IFSs, a generalization of C-IFSs, and develops a range of operations for their application MCDM. A novel method, CASPAS, is proposed, which builds on the Choquet integral and incorporates interactions between criteria, addressing a significant gap in traditional weighted aggregation methods. D-IFVs are ranked using score and accuracy functions, and algebraic operations are formulated by extending the radius limit from 1 to \(\sqrt{2}\), allowing the creation of weighted arithmetic and geometric aggregation operators. Additionally, disc intuitionistic fuzzy Choquet interval operators are introduced, which improve the ability to model complex interdependencies in real-world decision-making scenarios. The CASPAS method is applied to a renewable energy problem, specifically the selection of optimal solar panels, and demonstrates its effectiveness in improving decision-making accuracy by incorporating interdependencies among selection criteria. Comparative analyzes with existing methods, alongside sensitivity and validity analysis, confirm the robustness and applicability of the proposed method. The study’s contributions advance the theoretical framework of D-IFSs and provide practical tools for enhancing renewable energy decision-making, supporting the global transition to sustainable energy solutions. Future studies could focus on expanding the CASPAS method to other complex decision-making problems, especially in sectors such as healthcare, transportation, and environmental management. Additionally, further exploration of different Choquet aggregation operators and their real-world applications in MCDM could refine the proposed methodology and contribute to a more comprehensive understanding of decision-making under uncertainty.

\section*{Acknowledgments} 

We gratefully thank Dr. Abdullah Erkam Gündüz (Ankara Yıldırım Beyazıt University, Electrical-Electronics Engineering Department) and Research Assistant Muhammed Selman Erel (Ankara Yıldırım Beyazıt University, Electrical-Electronics Engineering Department) for providing the expert opinions.

\end{document}